\documentclass[journal, two column]{IEEEtran}
\usepackage{amsmath, amssymb, amsthm, bm}
\usepackage{cases}    % For numcases
\usepackage{amssymb}
\usepackage{graphicx}
\usepackage{color, colortbl} 
\usepackage[dvipsnames,svgnames,x11names]{xcolor}
\makeatother
\usepackage{tabularx, boldline, multirow}
\usepackage[boxed,ruled,vlined]{algorithm2e}
\usepackage{enumerate}
\usepackage[list=true]{subcaption}
\usepackage[font={footnotesize}]{caption}
\usepackage{cite}
\usepackage[normalem]{ulem}
\usepackage[shortlabels]{enumitem}
\usepackage{pgfplots}
\pgfplotsset{compat=1.17}
\usetikzlibrary{plotmarks}
\usepackage{tikz}
\usetikzlibrary {arrows.meta} 
\usepackage{pgfgantt}
\usepackage{mathtools}
\usepackage{accents}
\usepackage{balance}
\usepackage{acro}
\usepackage{arydshln}
\usepackage{scalerel}

\usepackage{titlesec}
\titlespacing*{\section}{0pt}{10pt plus 2pt minus 2pt}{2pt plus 2pt minus 2pt}
\titlespacing*{\subsection}{0pt}{5pt plus 1pt minus 1pt}{2pt plus 1pt minus 1pt}
\titlespacing*{\subsubsection}{0pt}{2pt plus 1pt minus 1pt}{1pt plus 1pt minus 1pt}

\newtheorem{proposition}{Proposition}

\usepackage[left=0.61 in,right=0.61 in,top=0.7 in, bottom=1 in]{geometry}
\DeclareAcronym{mle}{
short=MLE,
long= maximum likelihood estimation,
}
\DeclareAcronym{ils}{
short=ILS,
long= iterative least squares,
}
\DeclareAcronym{ml}{
short=ML,
long= maximum likelihood,
}
\DeclareAcronym{ap}{
short=AP,
long= access point,
}
\DeclareAcronym{nllf}{
short=NLLF,
long= negative log-likelihood function,
}
\DeclareAcronym{ue}{
short=UE,
long= user equipment,
}
\DeclareAcronym{iq}{
short=IQ,
long= in-phase and quadrature,
}
\DeclareAcronym{awgn}{
short=AWGN,
long= additive white Gaussian noise,
}
\DeclareAcronym{psd}{
short=PSD,
long= power spectral density,
}
\DeclareAcronym{lo}{
short=LO,
long= local oscillator,
}
\DeclareAcronym{mimo}{
short=MIMO,
long= multiple-input multiple-output,
}
\DeclareAcronym{mmimo}{
short=mMIMO,
long= massive multiple-input multiple-output,
}
\DeclareAcronym{toa}{
short=ToA,
long= time-of-arrival,
}
\DeclareAcronym{mse}{
short=MSE,
long= mean-squared error,
}
\DeclareAcronym{rmse}{
short=RMSE,
long= root-mean-squared error,
}
\DeclareAcronym{tdoa}{
short=TDoA,
long= time-difference-of-arrival,
}
\DeclareAcronym{gdop}{
short=GDOP,
long= geometric dilution of precision,
}
\DeclareAcronym{ofdm}{
short=OFDM,
long= orthogonal frequency division multiplexing,
}
\DeclareAcronym{ota}{
short=OtA,
long= over-the-air,
}
\DeclareAcronym{fr-c}{
short=FR-calibrated,
long= fully receive-calibrated,
}
\DeclareAcronym{ft-c}{
short=FT-calibrated,
long= fully transmit-calibrated,
}
\DeclareAcronym{f-c}{
short=F-calibrated,
long= fully calibrated,
}
\DeclareAcronym{los}{
short= LoS,
long=line-of-sight,
}
\DeclareAcronym{r-c}{
short=R-calibrated,
long= reciprocity calibrated,
}
\DeclareAcronym{crlb}{
short=CRLB,
long= Cramér-Rao lower bound,
}
\DeclareAcronym{fim}{
short=FIM,
long= Fisher information matrix,
}
\DeclareAcronym{aoa}{
short=AoA,
long= angle-of-arrival,
}
\DeclareAcronym{aod}{
short=AoD,
long= angle-of-departure,
}

\begin{document}
% \bstctlcite{IEEEexample:BSTcontrol}

\title{Location and Map-Assisted Wideband Phase and Time Calibration Between Distributed Antennas\\
\thanks{This work was supported by the Swedish Foundation for Strategic Research (SSF) under grant no. ID19-0021,  the Swedish Research Council under VR grants 2022-03007 and 2023-03414, ELLIIT, the KAW Foundation, and the Spanish Ministry of Science under grant PID2023-152820OB-I00. This work was also conducted within the Advanced Digitalization program at the WiTECH Centre DISCOURSE, financed by VINNOVA and partner companies. 
  }
}
\author{Yibo~Wu\IEEEauthorrefmark{1}\IEEEauthorrefmark{2},
Musa~Furkan~Keskin\IEEEauthorrefmark{2},
    Ulf~Gustavsson\IEEEauthorrefmark{1},\\
Gonzalo Seco-Granados\IEEEauthorrefmark{3},
Erik G. Larsson\IEEEauthorrefmark{4},
and Henk~Wymeersch\IEEEauthorrefmark{2}\\
        \IEEEauthorrefmark{1}Ericsson Research, Gothenburg, Sweden, 
        \IEEEauthorrefmark{2}%Department of Electrical Engineeering, 
        Chalmers University of Technology, Gothenburg, Sweden,\\
		\IEEEauthorrefmark{3}%Department of Telecommunications and Systems Engineering, 
Universitat Autonoma de Barcelona, Barcelona,
Spain, 
\IEEEauthorrefmark{4}%Department of Electrical Engineering (ISY), 
Link\"{o}ping University, Link\"{o}ping, Sweden
        % \\email: yibo@chalmers.se
        }
\maketitle
\thispagestyle{empty}

\begin{abstract}
Distributed \acl{mmimo} networks utilize a large number of distributed access points (APs) to serve multiple user equipments (UEs), offering significant potential for both communication and localization. However, these networks require frequent phase and time calibration between distributed antennas due to oscillator phase drifts, crucial for reciprocity-based coherent beamforming and accurate localization. While this calibration is typically performed through bi-directional measurements between antennas, it can be simplified to uni-directional measurement under perfect knowledge of antenna locations. This paper extends a recent phase calibration narrowband \ac{los} model to a phase and time calibration wideband \acl{ofdm} model, including both LoS and reflection paths and allowing for joint phase and time calibrations. We explore different scenarios, considering whether or not prior knowledge of antenna locations and the map is available. For each case, we introduce a practical \acl{ml} estimator and conduct \ac{crlb} analyses to benchmark performance. Simulations validate our estimators against the \ac{crlb} in these scenarios.
\end{abstract}
% \vspace{-5mm}
\acresetall 
\section{Introduction}
Distributed \ac{mmimo} is a promising technology for next-generation
communication systems, where numerous distributed \acp{ap} serve
multiple \acp{ue} with uniform
service~\cite{cell_free_book,ngo2017cell}. However, it requires
precise phase and time synchronization between \acp{ap}, critical for uplink combining, 
downlink MIMO beamforming, and   
localization~\cite{cell_free_book,ngo2017cell,sync_dist_2024,
  fascista2023uplink}. Bi-directional \ac{ota} measurements between the
antennas of different APs can be used to calibrate uplink-downlink reciprocity
phase errors for reciprocity-based
beamforming~\cite{vieira2017reciprocity}. Unless the \acp{lo}
in different APs are mutually locked to each other, the phase drift
originating from \ac{lo} noise requires frequent (millisecond-level)
re-calibration
\cite{xu2024distributed,nissel2022correctly,larsson2023phase}. Such
calibration can also be accomplished using bi-directional \ac{ota}
measurements \emph{between the
APs}~\cite{larsson2023phase,sync_dist_2024}.

%indicate that phase drift originating from the \ac{lo} requires more
%frequent, millisecond-level re-calibration using bi-directional
%\ac{ota} measurements~\cite{xu2024distributed}, rather than every few
%hours~\cite{bjornson2017massive}, revealing the misconception of phase
%calibration due to oscillator phase drift in the literature.

While~\cite{vieira2017reciprocity,nissel2022correctly,larsson2023phase,sync_dist_2024,xu2024distributed}
focus on phase calibration issues from hardware impairments and
asynchronous timing, time calibration is also crucial for delay-based
high-accuracy localization~\cite{fan2021carrier,fascista2023uplink,wymeersch2023fundamental}. Joint
phase and time calibration, especially addressing errors from
imperfect \acp{lo} between distributed antennas, remains
underexplored.  The paper~\cite{sync_dist_2024} offers a phase
calibration method for narrowband transmission over \ac{los} channels but does not address time calibration, wideband systems, and multi-path effects (which commonly arise from ground reflections~\cite{jaeckel2017explicit}). Additionally, the potential
benefits of using location (AP positions) and map information
(reflection points and phases) in phase and time calibration are not
systematically explored. Specifically, knowing the AP locations reduces the need for
calibration from bi-directional to uni-directional
measurements, and knowing the map information could enhance calibration
accuracy. These research gaps motivate our work.

\begin{figure}[t]
    \centering
    \input{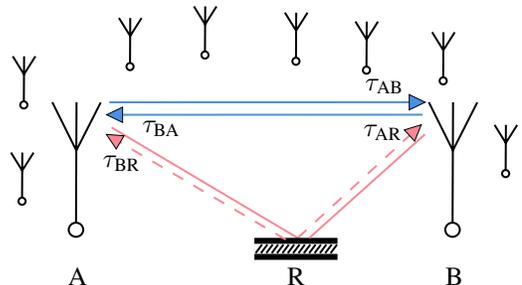}
    \caption{Phase and time calibration between two asynchronous distributed single-antenna APs in a distributed massive MIMO network, A and B, via LoS and ground reflection paths (reflection point R). Depending on the prior knowledge of the location (AP positions) and map (reflection point position and phase) information, bi-directional or uni-directional measurements are needed. }
    \label{fig:sys_model_AB}
\end{figure}

In this paper, we introduce a novel phase and time calibration signal
model for wideband distributed antenna systems, enabling joint
phase and time calibration for precise calibration and accurate
localization. We investigate various scenarios, considering the
availability of location (AP positions) and/or map information
(reflection position and phase), over both \ac{los} and reflection
paths. We derive \acf{ml}-based estimators for each scenario to solve
the calibration problem. Additionally, we provide a rigorous \ac{crlb}
analysis using Fisher information, establishing a theoretical
benchmark for performance evaluation. By examining a setup with two
single-antenna APs, we emphasize that our model and techniques
apply also to calibration between antennas of multi-antenna APs and between
APs and UEs in distributed \ac{mmimo} scenarios, broadening the potential
impact of our work.

%%%%%%%%%%%%%%%%%%%%%%%%%%%%%%%%%%%%%%%%%
% \section{Preliminaries}
% This section defines the phase and clock offsets and the antenna individual-calibration assumption.

\section{System Model}\label{section:sys_model_observ_model}
\subsection{Phase and Clock Offsets}\label{section:sys_model_phase_time_def}
Two single-antenna \acp{ap} A and B are shown in Fig.~\ref{fig:sys_model_AB}. Each \ac{ap}, A or B, has a transmit and receive branch, causing respective time and phase offsets~\cite{nissel2022correctly}. We assume that each \ac{ap} has already individually calibrated its transmit and receive antenna times and phases.\footnote{This assumption is valid as the antenna individual calibration can be done using electromagnetic simulations or anechoic chamber measurements~\cite{sevgi2008antenna}, or \ac{ota} methods~\cite{jiang2018framework,cao2023experimental}. More details on antenna individual calibration definition are in~\cite{larsson2023phase}.} The rest of this paper focuses on phase and time calibration between different antennas.

The propagation delay between A and B is 
\begin{align}
     \tau_{\text{AB}} &={\| \boldsymbol{p}_{\text{A}} - \boldsymbol{p}_{\text{B}}\|}/{c}, \label{eq:delay_define}
\end{align}
where $\boldsymbol{p}_{\text{A}}$ and $\boldsymbol{p}_{\text{B}}$ are the 2-D positions of A and B, and $c$ is the speed of light. Note $\tau_{\text{AB}} = \tau_{\text{BA}}$ due to reciprocity. %Assume a global time as an absolute reference. 
At global time zero, the clock offset for the transmission from A to B is $\delta_{t_{\text{AB}}} \in \mathbb{R}$, and from B to A is $\delta_{t_{\text{BA}}}\in \mathbb{R}$.\footnote{For simplicity, we assume that APs A and B transmit simultaneously at their local time zero, respectively. In practice, A and B can transmit at different local time, provided the time difference is pre-determined, similar to the assumption in the round-trip time protocol.} The additional phase offset transmitting from A to B and from B to A due to hardware impairments (e.g.,  \ac{iq} imbalance or common phase error due to phase noise) is denoted by $\delta_{\phi^{ }_{{\text{AB}}}} \in [0,2\pi)$ and $\delta_{\phi^{ }_{{\text{BA}}}} \in [0,2\pi)$. We note that
\begin{align}
    \delta_{t_{\text{AB}}} &= -\delta_{t_{\text{BA}}}, \label{eq:delta_t_AB_BA}\\
    \delta_{\phi_{\text{AB}}^{}} &= -\delta_{\phi_{\text{BA}}^{}}. \label{eq:delta_phi_AB_BA}
\end{align}
These minus signs arise because a positive transmit time delay shift has the opposite effect on the signal phase compared to a positive receive time delay shift, due to opposite transmission directions. See~\cite{nissel2022correctly,larsson2023phase} for more details. 

%%%%%%%%%%%%%%%%%%%%%%%%%%%%%%%%%%%%%%%%%%%%%%%%%%%%%%%
\vspace{-1mm}
\subsection{Observation Model via LoS Path}
We consider an OFDM system with subcarrier spacing $\Delta_f$ Hz and $N$ subcarriers. The bandwidth is $W=N\Delta_f$. The \ac{awgn} noise \ac{psd} is $N_0$. The transmit power is $P_{\text{tx}}$. The complex pilot symbol vector at one OFDM symbol, transmitted by AP $A$, is $\boldsymbol{s}_{\text{A}} = [s_{\text{A}}[0], \cdots, s_{\text{A}}[N-1]]^{\mathsf{T}} \in \mathbb{C}^{N}$. The average symbol energy $E_s = \mathbb{E}\{|{s}_{\text{A}}[n]|^2\} = P_{\text{tx}}/W$.

Under an ideal \ac{los} channel, the received symbol at AP B over subcarrier $n$, after filtering, sampling, and cyclic prefix removal, is given by\footnote{The model is equivalent to the model in~\cite[eq. (20)]{nissel2022correctly} by setting $\delta_{\phi^{ }_{{\text{AB}}}}=-(\varphi_t-\varphi_r)$ and $\tau_{\text{AB}}=\tau_t-\tau_r$. }
\begin{align}
    y_{\text{AB}}[n]
    % =\beta_{\text{AB}}  \times \nonumber
    % \\ &{ e^{-j({\phi^{ }_{t_{\text{A}}} - \phi^{ }_{r_{\text{B}}}})}}  {e^{-j2\pi (f_c+n\Delta_f) (\tau_{\text{AB}} + t_{\text{A}} - r_{\text{B}})}} s_{\text{A}}[n] + w_{\text{AB}}[n]
    % % _{\triangleq e^{-j\phi_{\text{AB}}}} + w_{\text{AB}} 
    % \\ 
    &= \beta_{\text{AB}} e^{-j \delta_{\phi^{ }_{{\text{AB}}}} } { e^{-j2\pi f_c {\tilde{\tau}_{\text{AB}}}}}  e^{-j2\pi n \Delta_f{\tilde{\tau}_{\text{AB}}}}   s_{\text{A}}[n] + w_{\text{AB}}[n] \notag \\
    &= \beta_{\text{AB}} e^{j {\varphi_{\text{AB}}}}
      {a(\tilde{\tau}_{\text{AB}})[n]} s_{\text{A}}[n] + w_{\text{AB}}[n],
    \label{eq:y_A_B}
\end{align}
where $\beta_{\text{AB}} \in \mathbb{R}$ is channel gain including path loss, $w_{\text{AB}}[n] \sim \mathcal{N}_{\mathbb{C}}(0,N_0)$ is \ac{awgn}.  The delay $\tau_{\text{AB}}$, pseudo-delay $\tilde{\tau}_{\text{AB}}$ (including the clock offset), carrier phase $\varphi_{\text{AB}}$, and delay steering vector element $[\boldsymbol{a}(\tilde{\tau}_{\text{AB}})]_n$, are 
\begin{align}
\tilde{\tau}_{\text{AB}} &= {\tau}_{\text{AB}} + \delta_{t_{\text{AB}}},  \\
    \varphi_{\text{AB}} & = -2\pi f_c \tilde{\tau}_{\text{AB}} -\delta_{\phi^{ }_{{\text{AB}}}}, \label{eq:carrier_phase_define} \\
    % \delta_{\varphi_{\text{AB}}}^{}&= -2\pi f_c\delta_{t_{\text{AB}}}\\
    [\boldsymbol{a}(\tilde{\tau}_{\text{AB}})]_n &=e^{-j2\pi n \Delta_f \tilde{\tau}_{\text{AB}}}.
\end{align}
where the delay steering vector $\boldsymbol{a}(\tilde{\tau}_{\text{AB}}) \triangleq [[\boldsymbol{a}(\tilde{\tau}_{\text{AB}})]_0,\cdots,[\boldsymbol{a}(\tilde{\tau}_{\text{AB}})]_{N-1}]^{\mathsf{T}} \in \mathbb{C}^{N}$. 
Collecting $y_{\text{AB}}[n]$ over $N$ subcarriers gives $\boldsymbol{y}_{\text{AB}} = [y_{\text{AB}}[0],\cdots,y_{\text{AB}}[{N-1}]]^{\mathsf{T}}$. We obtain~\eqref{eq:y_A_B} in vector form as
\begin{align}
    \boldsymbol{y}_{\text{AB}} 
    % & \beta_{\text{AB}} {e^{j{\varphi_{\text{AB}}}}}  \boldsymbol{a}(\tilde{\tau}_{\text{AB}}) \odot \boldsymbol{s}_{\text{A}} + \boldsymbol{w}_{\text{AB}} \nonumber \\
    =&{\boldsymbol{\mu}_{\text{AB}}} +\boldsymbol{w}_{\text{AB}},
    \label{eq:y_A_B_vector}
\end{align}
where $\boldsymbol{\mu}_{\text{AB}}\triangleq {\beta_{\text{AB}} {e^{j{\varphi_{\text{AB}}}}}  \boldsymbol{a}(\tilde{\tau}_{\text{AB}}) \odot \boldsymbol{s}_{\text{A}}}$, and $\boldsymbol{w}_{\text{AB}}=[{w}_{\text{AB}}[0],\cdots,{w}_{\text{AB}}[N-1]]^{\mathsf{T}}$. 

Similar to~\eqref{eq:y_A_B_vector}, we define the received signals at AP A, transmitted from AP B, as
\begin{align}
    \boldsymbol{y}_{\text{BA}} 
    % =& {\beta_{\text{BA}} {e^{j{\varphi_{\text{BA}}}}}  \boldsymbol{a}(\tilde{\tau}_{\text{BA}}) \odot \boldsymbol{s}_{\text{B}}} + \boldsymbol{w}_{\text{BA}} \nonumber \\
    =&{\boldsymbol{\mu}_{\text{BA}}} +\boldsymbol{w}_{\text{BA}}, \label{eq:y_B_A_vector}
\end{align}
where $\boldsymbol{\mu}_{\text{BA}} \triangleq {\beta_{\text{BA}} {e^{j{\varphi_{\text{BA}}}}}  \boldsymbol{a}(\tilde{\tau}_{\text{BA}}) \odot \boldsymbol{s}_{\text{B}}}$. The pseudo delay $\tilde{\tau}_{\text{BA}}$ and carrier phase $\varphi_{\text{BA}}$ are 
\begin{align}
\tilde{\tau}_{\text{BA}} &= {\tau}_{\text{BA}} + \delta_{t_{\text{BA}}},\\
 \varphi_{\text{BA}}&= -2\pi f_c \tilde{\tau}_{\text{BA}} - \delta_{\phi_{\text{BA}}^{}} 
 % \\
% [\boldsymbol{a}(\tilde{\tau}_{\text{BA}})]_n &=e^{-j2\pi n \Delta_f \tilde{\tau}_{\text{BA}}}.
\end{align}
Note that $\beta_{\text{AB}}=\beta_{\text{BA}}$ and $\tau_{\text{AB}}=\tau_{\text{BA}}$ due to reciprocity.

\subsection{Observation Model via LoS and Reflection Paths}
In practice, a pure LoS path from array A to B is not always guaranteed due to multi-path effects like ground reflections. The observation from array A to B, involving a LoS path and a reflection path reflected at a reflection $R$ is represented by
\begin{align}
    \boldsymbol{y}_{\text{AB}} 
    % &= {\boldsymbol{\mu}_{\text{AB}}} + {\beta_\text{AR} {e^{j{\varphi_\text{AR}}}}  \boldsymbol{a}(\tilde{\tau}_\text{AR}) \odot \boldsymbol{s}_{\text{A}}} + \boldsymbol{w}_{\text{AB}} \\
    &={\boldsymbol{\mu}_{\text{AB}}} +\boldsymbol{\mu}_\text{AR}+\boldsymbol{w}_{\text{AB}}, \label{eq:y_AB_1R}
\end{align}
where ${\boldsymbol{\mu}_\text{AR}}\triangleq{\beta_\text{AR} {e^{j{\varphi_\text{AR}}}}  \boldsymbol{a}(\tilde{\tau}_\text{AR}) \odot \boldsymbol{s}_{\text{A}}}$. The pseudo-delay, $\tilde{\tau}_{\text{AR}}$, and carrier phase, $\varphi_\text{AR}$, of the reflection path are defined as
\begin{align}
 \tilde{\tau}_{\text{AR}}&= {\tau}_\text{AR} + \delta_{t_{\text{AB}}},\\
  \varphi_\text{AR}&= -2\pi f_c \tilde{\tau}_\text{AR} - \delta_{\phi_\text{AR}^{ }} - \delta_{\phi^{ }_{\text{AB}}}.
%[\boldsymbol{a}(\tilde{\tau}_{\text{AB}})]_{n}&=e^{-j2\pi n \Delta_f \tilde{\tau}_\text{AR}},
\end{align}
Here the reflection path delay $\tau_{\text{AR}}$ is  given by $\tau_{\text{AR}} = ({||\boldsymbol{p}_\text{A} - \boldsymbol{p}_{\text{R}} || + ||\boldsymbol{p}_{\text{R}} - \boldsymbol{p}_\text{B} ||})/{c}$, $\boldsymbol{p}_{\text{R}}$ is the reflection location, and the reflection introduces an unknown phase rotation $\delta_{\phi_\text{AR}^{ }}$.

Similarly, we can rewrite the observations from AP B to A~\eqref{eq:y_A_B_vector}, via a LoS and reflection paths, as
\begin{align}
    \boldsymbol{y}_{\text{BA}} &= {\boldsymbol{\mu}_{\text{BA}}} + {\boldsymbol{\mu}_\text{BR}} + \boldsymbol{w}_{\text{BA}}, \label{eq:y_BA_1R}
\end{align}
where ${\boldsymbol{\mu}_\text{BR}}={\beta_\text{BR} {e^{j{\varphi_\text{BR}}}}  \boldsymbol{a}(\tilde{\tau}_\text{BR}) \odot \boldsymbol{s}_{\text{B}}}$, and the pseudo-delay and carrier phase of the reflection path transmitted from AP B are defined as
\begin{align}
 \tilde{\tau}_\text{BR}&= {\tau}_\text{BR} + \delta_{t_{\text{BA}}} \\
  \varphi_\text{BR}&= -2\pi f_c \tilde{\tau}_\text{BR} - \delta_{\phi_\text{BR}}^{ } - \delta_{\phi^{ }_{\text{BA}}}.
\end{align}
Note that $\beta_\text{BR}=\beta_\text{AR}$, the phase rotation $\delta_{\phi_{\text{BR}}}^{ }=\delta_{\phi_\text{AR}^{ }}$, and ${\tau}_\text{BR}={\tau}_\text{AR}$ due to reciprocity~\cite{zhang2020effects}.

\subsection{Problem Formulation}\label{section:prob_form}
The task is to jointly calibrate phase and time between two single-antenna APs using uni-directional (known AP positions) or bi-directional (unknown AP positions) \ac{ota} observations. Two scenarios are considered based on the knowledge of AP positions.   

\begin{itemize}[noitemsep, topsep=0pt,leftmargin=*]
    \item \textbf{Known AP positions}: With known positions of APs A and B, we can determine the delay $\tau_{\text{AB}}$, which helps to calibrate the phase and clock offsets using uni-directional observation over LoS and reflection paths (if the paths are resolvable), e.g., $\boldsymbol{y}_{\text{AB}}$ from~\eqref{eq:y_AB_1R}. With \textit{unknown map}, i.e., unknown $\tau_{\text{AR}}$ or $\delta_{\phi_\text{AR}^{ }}$, and a uni-directional \textit{two-path} observation~\eqref{eq:y_AB_1R}, the unknown parameter vector is defined as 
\begin{align}
    \boldsymbol{\eta}_{}=[\delta_{t_{\text{AB}}}, \delta_{\phi_{\text{AB}}^{}}, \tau_{\text{AR}}, \delta_{\phi_\text{AR}^{ }},\beta_{\text{AB}}, \beta_\text{AR}
    ]^{\mathsf{T}} \in \mathbb{R}^{6 \times 1} ~.\label{eq:unknownPar_known_positions}
\end{align}
If the reflection position $\boldsymbol{p}_{\text{R}}$ or the reflection rotation phase $\delta_{\phi_\text{AR}^{ }}$ is known, i.e., \textit{known map}, the unknown vector $\boldsymbol{\eta}$ in~\eqref{eq:unknownPar_known_positions} is reduced by excluding $\tau_{\text{AR}}$ or $\delta_{\phi_\text{AR}^{ }}$, respectively. 

With only a \textit{LoS} path observation, the unknown parameter vector in~\eqref{eq:unknownPar_known_positions} is reduced to
\begin{align}
    \boldsymbol{\eta}_{}=[\delta_{t_{\text{AB}}}, \delta_{\phi_{\text{AB}}^{}}, \beta_{\text{AB}}
    ]^{\mathsf{T}} \in \mathbb{R}^{3 \times 1}. \label{eq:unknownPar_known_positions_LoS}
\end{align}

\item \textbf{Unknown AP positions}: Without knowing the positions of APs A and B, a uni-directional observation, e.g., $\boldsymbol{y}_{\text{AB}}$ from~\eqref{eq:y_AB_1R}, is insufficient to independently estimate both $\tau_{\text{AB}}$ and $\delta_{t_{\text{AB}}}$, because there are more unknowns ($6$ in \eqref{eq:unknownPar_known_positions} and $\tau_{\text{AB}}$) than measurable parameters ($ \varphi_{\text{AB}}, \varphi_{\text{AR}}, \tilde{\tau}_{\text{AB}}, \tilde{\tau}_{\text{AR}}, \beta_{\text{AB}},\beta_{\text{AR}}$). Instead, bi-directional observations $\boldsymbol{y}_{\text{AB}}$ and $\boldsymbol{y}_{\text{BA}}$ are needed. With \textit{unknown map} and bi-directional \textit{two-path} observations, the unknown parameter vector is defined as
\begin{align}
    \boldsymbol{\eta}_{}=[\tau_{\text{AB}},\delta_{t_{\text{AB}}}, \delta_{\phi_{\text{AB}}^{}}, \tau_{\text{AR}}, \delta_{\phi_\text{AR}^{ }},\beta_{\text{AB}}, \beta_\text{AR}
    ]^{\mathsf{T}} \in \mathbb{R}^{7 \times 1} ~. \label{eq:unknownPar_unknown_positions}
\end{align}
If the reflection information $\tau_{\text{AR}}$ and $\delta_{\phi_{\text{AR}}}$ are known, i.e., \textit{known map}, the unknown vector $\boldsymbol{\eta}$ in~\eqref{eq:unknownPar_unknown_positions} is reduced by excluding $\tau_{\text{AR}}$ or $\delta_{\phi_\text{AR}^{ }}$, respectively. 

With only \textit{LoS} path observations, i.e., using observations $\boldsymbol{y}_{\text{AB}}^{\mathsf{}}$ from~\eqref{eq:y_A_B_vector} and $\boldsymbol{y}_{\text{BA}}^{\mathsf{}}$~\eqref{eq:y_B_A_vector}, the unknown parameter vector is 
\begin{align}
    \boldsymbol{\eta}_{}=[\tau_{\text{AB}}, \delta_{t_{\text{AB}}}, \delta_{\phi_{\text{AB}}^{}}, \beta_{\text{AB}}
    ]^{\mathsf{T}} \in \mathbb{R}^{4 \times 1}. \label{eq:unknownPar_unknown_positions_LoS}
\end{align} 
\end{itemize}

\section{Phase and Time Calibration}% of Two Single-Antenna APs}
This section introduces four novel \ac{ml}-based estimators for the phase and time calibration tasks~\eqref{eq:unknownPar_known_positions},~\eqref{eq:unknownPar_known_positions_LoS}, \eqref{eq:unknownPar_unknown_positions}, and~\eqref{eq:unknownPar_unknown_positions_LoS} in Section~\ref{section:prob_form}. These estimators share a similar form but differ in their final grid search dimensions.

\begin{proposition}
The \ac{ml}-based estimator of the unknown parameters $\boldsymbol{\eta}$ in \eqref{eq:unknownPar_known_positions}, \eqref{eq:unknownPar_known_positions_LoS}, \eqref{eq:unknownPar_unknown_positions}, and \eqref{eq:unknownPar_unknown_positions_LoS} are respectively given by 
\begin{align}
    \hspace{-1cm}[\hat{\delta}_{t_{\text{AB}}}, \hat{\tau}_{\text{AR}}, \hat{\delta}_{\phi_\text{AR}^{ }}] &=  \arg \underset{\substack{\delta_{t_{\text{AB}}}, \tau_{\text{AR}}, \\ \delta_{\phi_\text{AR}^{ }}}}{\min} \mathcal{L}_{\text{2-path}}^{\text{Uni}} (\delta_{t_{\text{AB}}}, \tau_{\text{AR}}, \delta_{\phi_\text{AR}^{ }}), 
   \label{eq:1way_2path_Estimator} \\
   \hat{\delta}_{t_{\text{AB}}} &=  \arg \underset{\delta_{t_{\text{AB}}}}{\min } \mathcal{L}_{\text{LoS}}^{\text{Uni}} (\delta_{t_{\text{AB}}}), 
   \label{eq:1way_LoS_Estimator} \\
 [\hat{\tau}_{\text{AB}}, \hat{\delta}_{t_{\text{AB}}},\hat{\tau}_{\text{AR}}, \hat{\delta}_{\phi_\text{AR}^{ }}] &=  \arg \underset{\substack{{\tau}_{\text{AB}}, \delta_{t_{\text{AB}}},\\ {\tau}_{\text{AR}}, {\delta}_{\phi_\text{AR}^{ }}}}{\min} \mathcal{L}_{\text{2-path}}^{\text{Bi}} ({\tau}_{\text{AB}}, \delta_{t_{\text{AB}}},{\tau}_{\text{AR}}, {\delta}_{\phi_\text{AR}^{ }}),\label{eq:2way_2path_Estimator}  \\
 [\hat{\tau}_{\text{AB}}, \hat{\delta}_{t_{\text{AB}}}] &=  \arg \underset{{\tau}_{\text{AB}}, \delta_{t_{\text{AB}}}}{\min} \mathcal{L}_{\text{LoS}}^{\text{Bi}} ({\tau}_{\text{AB}}, \delta_{t_{\text{AB}}}),
 \label{eq:2way_LoS_Estimator} 
\end{align}
where \eqref{eq:1way_2path_Estimator} and \eqref{eq:1way_LoS_Estimator} are the two-path and LoS estimators using the uni-directional observation $\boldsymbol{y}_{\text{AB}}$ from~\eqref{eq:y_AB_1R} and~\eqref{eq:y_A_B_vector}, respectively.  \eqref{eq:2way_2path_Estimator} and \eqref{eq:2way_LoS_Estimator} are the two-path and LoS estimators using the bi-directional observations $\boldsymbol{y}_{\text{}} =[\boldsymbol{y}_{\text{AB}}^{\mathsf{T}}, \boldsymbol{y}_{\text{BA}}^{\mathsf{H}}]^{\mathsf{T}}$, obtained from~\eqref{eq:y_AB_1R} and~\eqref{eq:y_BA_1R}, and from~\eqref{eq:y_A_B_vector} and~\eqref{eq:y_B_A_vector}, respectively.
The \acp{nllf} $\mathcal{L}_{\text{2-path}}^{\text{Uni}}$, $\mathcal{L}_{\text{LoS}}^{\text{uni}}$, $\mathcal{L}_{\text{2-path}}^{\text{Bi}}$, and $\mathcal{L}_{\text{2-path}}^{\text{Bi}}$  share a similar form as
\begin{align}
    \mathcal{L}_{\text{ML}}(\cdot) = &\left\| \ddot{\boldsymbol{y}}_{\text{}}\right\|^2_2 + \left\| \ddot{\boldsymbol{c}}_{\text{}}\right\|^2_2
   - 2 \left| \ddot{\boldsymbol{c}}_{\text{}}^{\mathsf{H}} \ddot{\boldsymbol{y}}_{\text{}}\right|,  \label{eq:NLLF_general}
\end{align}
where, based on the specific estimator, $\mathcal{L}_{\text{ML}} =\{\mathcal{L}_{\text{2-path}}^{\text{Uni}}, \mathcal{L}_{\text{LoS}}^{\text{Uni}}, \mathcal{L}_{\text{2-path}}^{\text{Bi}}, \mathcal{L}_{\text{LoS}}^{\text{Bi}} \}$, $\ddot{\boldsymbol{y}}$ and $\ddot{\boldsymbol{c}}$ in \eqref{eq:NLLF_general} are substituted accordingly to obtain the respective NLLF functions in \eqref{eq:2Path_est_knownPos}, \eqref{eq:NLLF_1way_LoS_Estimator}, \eqref{eq:2Path_est_unknownPos}, and \eqref{eq:NLLF_1Path_est_unknownPos}. 

\end{proposition}
\begin{proof}
    The detailed derivation of \eqref{eq:2way_2path_Estimator} is given in Appendix~\ref{sec:Appendix_TwoPath_unknown}. The derivation of other three estimators follows similar steps, given in Appendix~\ref{sec:Appendix_Other3Estimators}.
\end{proof}

Analytically solving \eqref{eq:1way_2path_Estimator}, \eqref{eq:1way_LoS_Estimator}, \eqref{eq:2way_2path_Estimator}, and \eqref{eq:2way_LoS_Estimator} is infeasible. Thus, a practical approach is to perform a grid search over the 3-D, 1-D, 4-D, and 2-D parameter spaces for \eqref{eq:1way_2path_Estimator}, \eqref{eq:1way_LoS_Estimator}, \eqref{eq:2way_2path_Estimator}, and \eqref{eq:2way_LoS_Estimator}, respectively.

\section{Cramér-Rao Lower Bound and Analysis}
Next, we derive the \ac{crlb} for parameters in~\eqref{eq:unknownPar_known_positions},~\eqref{eq:unknownPar_known_positions_LoS}, \eqref{eq:unknownPar_unknown_positions},  and~\eqref{eq:unknownPar_unknown_positions_LoS}, respectively. The \ac{crlb} on the error variance of any unbiased estimator of the unknown parameters $\boldsymbol{\eta}$ is defined as
\begin{align}
    \mathbb{E}_{\boldsymbol{\eta}}\{(\hat{\boldsymbol{\eta}}-\boldsymbol{\eta}) (\hat{\boldsymbol{\eta}}-\boldsymbol{\eta})^{\mathsf{T}}
    \}\geq (\boldsymbol{J}_{\boldsymbol{\eta}_{\text{}}})^{-1}. \label{eq:crlb_var_definition}
\end{align}
The \ac{fim} of $\boldsymbol{\eta}$ can be calculated as
\begin{align}
[\boldsymbol{J}_{\boldsymbol{\eta}_{\text{}}}]_{i,j} &= 2\Re\Big\{\frac{\partial \boldsymbol{\mu}_{}^{\mathsf{H}}}{\partial [\boldsymbol{\eta}_{\text{}}]_i } \frac{\partial \boldsymbol{\mu}_{}^{\mathsf{}}}{\partial [\boldsymbol{\eta}_{\text{}}]_j} \Big\} \frac{1}{N_0}, \label{eq:FIM_Compute}
\end{align}
where $\boldsymbol{\mu}$ is chosen as  $[\boldsymbol{\mu}_{\text{AB}}^{\mathsf{T}}+\boldsymbol{\mu}_\text{AR}^{\mathsf{T}}]^{\mathsf{T}}$, $\boldsymbol{\mu}_{\text{AB}}$, and $[\boldsymbol{\mu}_{\text{AB}}^{\mathsf{T}},\boldsymbol{\mu}_\text{BA}^{\mathsf{H}}]^{\mathsf{T}}$ for the unknown parameters $\boldsymbol{\eta}$ from~\eqref{eq:unknownPar_known_positions},~\eqref{eq:unknownPar_known_positions_LoS}, and~\eqref{eq:unknownPar_unknown_positions_LoS}, respectively. 
Using the \ac{fim} for an unknown parameter vector $\boldsymbol{\eta}$, we calculate the lower bound on the error variance for the $i$-th parameter as
\begin{align}
     \text{var}(\hat{\eta}_{i}) \geq [\boldsymbol{J}_{\boldsymbol{\eta}_{\text{}}}^{-1}]_{i,i}. \label{eq:each_error_bound_general}
\end{align}

\subsection{Known Positions with LoS Path or Two-path}
\subsubsection{LoS Path}
In the scenario with known AP positions and uni-directional observation~\eqref{eq:y_A_B_vector} over LoS path, the analytical \acp{crlb} of the error variance of the unknown $\hat{\delta}_{t_{\text{AB}}}$ and $\hat{\delta}_{\phi^{ }_{\text{AB}}}$ from~\eqref{eq:unknownPar_known_positions_LoS} are 
\begin{align}
   \text{var}(\hat{\delta}_{t_{\text{AB}}})&= [\boldsymbol{J}_{\boldsymbol{\eta}_{\text{}}}^{-1}]_{1,1} = 3/(2\pi^2 W^2 \text{SNR}) \label{eq:CRB_1Path_delta_tau_AB}, \\
\text{var}(\hat{\delta}_{\phi^{ }_{\text{AB}}})&= [\boldsymbol{J}_{\boldsymbol{\eta}_{\text{}}}^{-1}]_{2,2} = 6f_c^2/(W^2\text{SNR}) + 1/(2\text{SNR}),
\label{eq:CRB_1Path_delta_phi_AB}
\end{align}
where the $\text{SNR} \triangleq E_s \beta^2N/N_0$. We omit the derivation details for simplicity. The CRLB for $\hat{\delta}_{t_{\text{AB}}}$~\eqref{eq:CRB_1Path_delta_tau_AB} shows that increasing the bandwidth can reduce the estimation error to zero. In contrast, the CRLB for $\hat{\delta}_{\phi^{ }_{\text{AB}}}$~\eqref{eq:CRB_1Path_delta_phi_AB} indicates a performance threshold. Referring to Eq.~\eqref{eq:carrier_phase_define},  this occurs because, with perfectly known $\delta_{t_{\text{AB}}}$, the error in the estimation of $\delta_{\phi^{ }_{{\text{AB}}}}$ is independent of bandwidth and only depends on SNR.

\subsubsection{Two-path}
In the scenario with known AP positions and uni-directional observation~\eqref{eq:y_AB_1R} over two paths (LoS and reflection), the \acp{crlb} of the error variance of unknown parameters in $\boldsymbol{\eta}$ from~\eqref{eq:unknownPar_known_positions} is given in Appendix~\ref{appdices: CRLB_1}. 

\subsection{Unknown Positions with LoS Path or Two-path} \label{section:CRLB_unknownPos}
\subsubsection{LoS Path}
The analytical CRLBs of the error variance of $\hat{\tau}_{\text{AB}}$, $\hat{\delta}_{t_{\text{AB}}}$, and $\hat{\delta}_{\phi^{ }_{\text{AB}}}$, from~\eqref{eq:unknownPar_unknown_positions_LoS}, are
\begin{align}
\text{var}(\hat{\tau}_{\text{AB}})&\geq [\boldsymbol{J}_{\boldsymbol{\eta}_{\text{}}}^{-1}]_{1,1} = 1/\big((16f_c^2 +{4 W^2}/{3}  )\pi^2\text{SNR}\big) \label{eq:CRB_1Path_2Obv_tau_AB} \\
   \text{var}(\hat{\delta}_{t_{\text{AB}}})&\geq [\boldsymbol{J}_{\boldsymbol{\eta}_{\text{}}}^{-1}]_{2,2} = 3/(4\pi^2 W^2 \text{SNR}), \label{eq:CRB_1Path_2Obv_delta_tau_AB} \\
\text{var}(\hat{\delta}_{\phi^{ }_{\text{AB}}})&\geq [\boldsymbol{J}_{\boldsymbol{\eta}_{\text{}}}^{-1}]_{2,2} = 3f_c^2/(W^2\text{SNR}) + 1/(4\text{SNR}). \label{eq:CRB_1Path_2Obv_delta_phi_AB}
\end{align}
We omit the derivation details for simplicity.
Comparing~\eqref{eq:CRB_1Path_2Obv_delta_tau_AB} and~\eqref{eq:CRB_1Path_2Obv_delta_phi_AB} with the uni-directional CRLBs,~\eqref{eq:CRB_1Path_delta_tau_AB} and~\eqref{eq:CRB_1Path_delta_phi_AB}, bi-directional observations halve the estimation variances of ${\delta}_{t_{\text{AB}}}$ and ${\delta}_{\phi^{ }_{\text{AB}}}$. Specifically,~\eqref{eq:CRB_1Path_2Obv_tau_AB} shows that the estimation variance of ${\tau}_{\text{AB}}$ decreases with $f_c^2$, unlike $\delta_{t_{\text{AB}}}$ and $\delta_{\phi^{ }_{{\text{AB}}}}$. This is because ${\delta}_{t_{\text{AB}}}$ and $\delta_{\phi^{ }_{{\text{AB}}}}$ can somehow be canceled out due to their opposite signs in the carrier phases of bi-directional measurements: i.e., $-{\delta}_{\phi^{ }_{\text{AB}}}$ in $\varphi_{\text{AB}}$ and $+ {\delta}_{\phi^{ }_{\text{AB}}}$ in $\varphi_{\text{BA}}$. The same applies to ${\delta}_{t^{ }_{\text{AB}}}$. This in turn improves the estimation of $\tau_{\text{AB}}$ due to its same signs in $\varphi_{\text{AB}}$ and $\varphi_{\text{BA}}$.
\subsubsection{Two-path}
In the scenario with unknown AP positions and bi-directional observation (\eqref{eq:y_AB_1R} and \eqref{eq:y_BA_1R}) over two paths (LoS and reflection), the \acp{crlb} of the error variance of unknown parameters in $\boldsymbol{\eta}$ from~\eqref{eq:unknownPar_unknown_positions} is given in Appendix~\ref{appdices: CRLB_2}.
%%%%%%%%%%%%%%%%%%%%%%%%%%%%%%%%%%%%%%%%
\newcolumntype{Y}{>{\centering\arraybackslash}X}
\newcolumntype{C}[1]{>{\centering\arraybackslash}p{#1}}
% \newcolumntype{Y1}{>{\centering\arraybackslash}X[2]}
\begin{table}[t]
\centering
\caption{Simulation parameters.}
\label{tab:dl_params}
\begin{tabularx}{0.91\linewidth}{p{0.24\linewidth}|C{0.65\linewidth}}
\hline
\textbf{Parameter} & \textbf{Value} \\ \hline
$\boldsymbol{p}_{\text{A}},\boldsymbol{p}_{\text{B}},\boldsymbol{p}_{\text{R}}$ & $[50,50],[0,0],[0,-10]$ m \\ \hline
$P_\text{tx}$, \(N_0 \), $f_c$, $\Delta_f$ & $10$ mW, $-174$ dBm/Hz, $2$ GHz, $60$ kHz \\ \hline
$\delta_{t_{\text{AB}}}$, \(\delta_{\phi_{\text{AB}}^{}}\), $\delta_{\phi_\text{AR}^{ }}$ & $0.67 \mu$s, $10^\circ$, $20^\circ$  \\ \hline
\shortstack{$\beta_{\text{AB}}$,  $\beta_{\text{AR}}$} & $\lambda/({4 \pi \left\|\boldsymbol{p}_\text{A} - \boldsymbol{p}_\text{B}\right\|})$, $\lambda/({4 \pi \left\| |\boldsymbol{p}_\text{A} - \boldsymbol{p}_\text{R}| + |\boldsymbol{p}_\text{R}-\boldsymbol{p}_\text{B}|\right\|})$ \\ \hline
\end{tabularx} \label{tab:Scenario_par}
\end{table}
%%%%%%%%%%%%%%%%%%%%%%%%%%%%%%%%%%%%%%%%
\begin{table}[t]
\centering
\caption{Setups of two scenarios and the corresponding estimators.}
\label{tab:simulation_params}
\begin{tabularx}{\linewidth}{p{0.3\linewidth}|Y|Y}
\hline
\textbf{Setups} & \textbf{Scenario 1} & \textbf{Scenario 2} \\ \hline
Prior Localization & Known $\boldsymbol{p}_{\text{A}}$, $\boldsymbol{p}_{\text{B}}$ & Unknown $\boldsymbol{p}_{\text{A}}$, $\boldsymbol{p}_{\text{B}}$\\ \hline
Unknown parameters $\boldsymbol{\eta}$ & Two-path: \eqref{eq:unknownPar_known_positions}, \hspace{0.3cm} LoS: \eqref{eq:unknownPar_known_positions_LoS}
   &Two-path:\eqref{eq:unknownPar_unknown_positions}, LoS:\eqref{eq:unknownPar_unknown_positions_LoS}  \\ \hline
Prior Map in two-path & Reduced $\boldsymbol{\eta}$ \eqref{eq:unknownPar_known_positions}  & Reduced $\boldsymbol{\eta}$ \eqref{eq:unknownPar_unknown_positions}  \\ \hline
Estimators & Two-path est.:\eqref{eq:1way_2path_Estimator}, LoS est.:\eqref{eq:1way_LoS_Estimator} &  LoS est.:\eqref{eq:2way_LoS_Estimator}\\ \hline
\end{tabularx} \label{tab:scen_details}
\end{table}
\vspace{-4mm}
\section{Simulation Results}
\subsection{Scenarios}
We consider a simulation setup involving two single-antenna APs. The setup parameters, following similar parameters in~\cite{fascista2023uplink}, are given in Table~\ref{tab:Scenario_par}.
% The reflection point is at $\boldsymbol{p}_{\text{R}}=[0,-10]$ m. The APs transmit with a power of $P_\text{tx}=10$ mW at a carrier frequency of \(f_c = 2 \text{ GHz}\), and a subcarrier spacing of \(\Delta f = 60\text{ kHz}\). The clock and additional phase offsets between the two APs are set as \(\delta_{t_{\text{AB}}} = -\delta_{t_{\text{BA}}} =0.67 \mu \text{s } (\approx 200 \text{ m}/c) \) and \(\delta_{\phi_{\text{AB}}^{}} = -\delta_{\phi_{\text{BA}}^{}} = 10^\circ\), respectively, following the parameters used in~\cite{fascista2023uplink}. The phase rotation caused by the reflection path is $\delta_{\phi_\text{AR}^{ }}=\delta_{\phi_\text{BR}^{ }}=20^\circ$.  Channel amplitudes are generated based on the common path loss model in free space as $
% \beta_{\text{AB}} = \lambda/({4 \pi \left\|\boldsymbol{p}_\text{A} - \boldsymbol{p}_\text{B}\right\|})$ and $\beta_{\text{AR}} = \lambda/({4 \pi \left\| |\boldsymbol{p}_\text{A} - \boldsymbol{p}_\text{R}| + |\boldsymbol{p}_\text{R}-\boldsymbol{p}_\text{B}|\right\|})$, where \(\lambda=c/f_c\) is the wavelength. The noise \ac{psd} is set to \(N_0 = -174 \, \text{dBm/Hz}\). 
The evaluated scenarios and estimators are given in Table~\ref{tab:scen_details}. The \ac{ml} estimator~\eqref{eq:2way_2path_Estimator} is not evaluated due to its impractical and costly 4-D grid search.

%%%%%%%%%%%%%%%%%%%%%%%%%%%%%%%%%%%%%%%%%%%
\subsection{Scenario 1: Calibration with Known AP Positions}
\begin{figure}[t]
    \centering
    \input{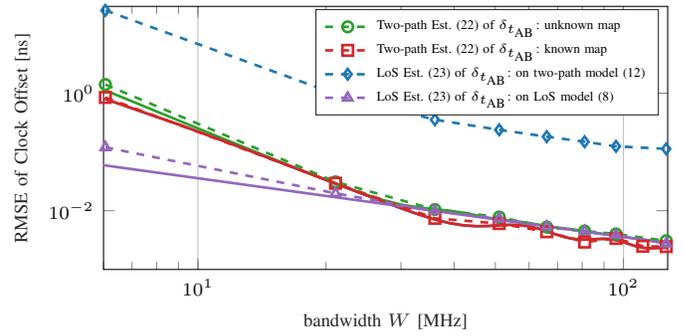}
    \caption{RMSE of the clock offset estimate $\delta_{t_{\text{AB}}}$ in ns versus the bandwidth for various estimators for the scenario with known AP positions. Solid lines are the corresponding CRLBs. }
    \label{fig:ClockOff_vs_BW_knownPos}
    \vspace{-0.3cm}
\end{figure}
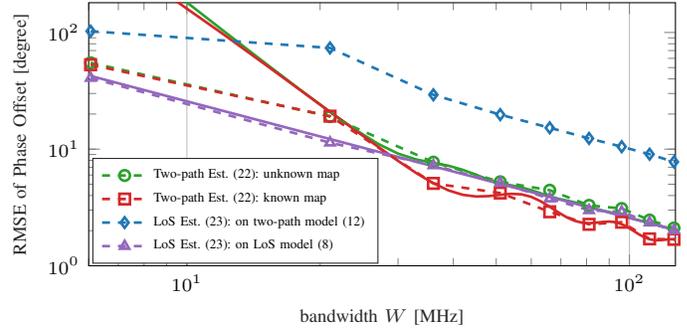
\begin{figure}[t]
    \centering
    \begin{tikzpicture}[font=\scriptsize]
\definecolor{color2}{rgb}{0.12156862745098,0.466666666666667,0.705882352941177}
\definecolor{color0}{rgb}{1,0.498039215686275,0.0549019607843137}
\definecolor{color1}{rgb}{0.172549019607843,0.627450980392157,0.172549019607843}
\definecolor{color3}{rgb}{0.83921568627451,0.152941176470588,0.156862745098039}
\definecolor{color4}{rgb}{0.580392156862745,0.403921568627451,0.741176470588235}
\definecolor{color5}{rgb}{0.549019607843137,0.337254901960784,0.294117647058824}
\definecolor{color6}{rgb}{0.890196078431372,0.466666666666667,0.76078431372549}
\definecolor{color7}{rgb}{0.737254901960784,0.741176470588235,0.133333333333333}

\begin{axis}[%
width=7.8cm,
height=3.5 cm,
at={(0,0)},
scale only axis,
xmode=log,
xmin=6,
xmax=127.06,
xlabel style={font=\color{white!15!black}},
xlabel={\scriptsize bandwidth $W$ [MHz] },
ymode=log,
ymin=1,
ymax=180,
ylabel style={font=\color{white!15!black}},
ylabel={\scriptsize RMSE of Phase Offset [degree]},
axis background/.style={fill=white},
title style={font=\bfseries},
xmajorgrids,
%xminorgrids,
%minor xtick={1, 2, 3, 4, 5, 6, 7, 8, 9, 10,20,30,40,50},
%minor x tick num=9,
%ymajorgrids,
legend style={at={(0.25,0.41)}, font=\tiny,anchor=north, legend cell align=left, draw=white!15!black},
legend columns=1,
%ytick distance={2},
%xtick distance={10},
%xtick={1,10,20,50}, % Custom tick positions
%xticklabels={$10^0$,$10^1$,20,50} % Custom tick labels]
]

\addplot [color=color1, line width=1.0pt, forget plot]
  table[row sep=crcr]{%
6.06  803.807484617477\\
9.06  242.71091634502\\
12.06  104.167248768759\\
15.06  54.3303254781388\\
18.06  32.1436558766646\\
21.06  20.8545007524681\\
24.06  14.6265101700315\\
27.06  11.0827599546318\\
30.06  9.12449772057189\\
33.06  8.13205446549055\\
36.06  7.62622180765562\\
39.06  7.19781766720414\\
42.06  6.63911214151283\\
45.06  6.01521719961128\\
48.06  5.48668037821005\\
51.06  5.12628636409749\\
54.06  4.89223381868108\\
57.06  4.68796575426597\\
60.06  4.44680193659436\\
63.06  4.17957207816283\\
66.06  3.93986277748013\\
69.06  3.76030600542634\\
72.06  3.62819872594257\\
75.06  3.50517049430266\\
78.06  3.36503680234655\\
81.06  3.21473545018542\\
84.06  3.07927615736055\\
87.06  2.97308447915461\\
90.06  2.88823274364633\\
93.06  2.8051463978296\\
96.06  2.71241085458404\\
99.06  2.61595921454061\\
102.06  2.52959545335758\\
105.06  2.45976758496585\\
108.06  2.40034202642097\\
111.06  2.34005085431785\\
114.06  2.27391497673075\\
117.06  2.20699306347218\\
120.06  2.14753218242083\\
123.06  2.09814926861498\\
126.06  2.05390636506006\\
129.06  2.00795445329896\\
132.06  1.9584275191318\\
135.06  1.90951995367592\\
138.06  1.86630704779927\\
141.06  1.82945629247809\\
144.06  1.7950031557511\\
147.06  1.75872724051946\\
150.06  1.7203437272721\\
};
%\addlegendentry{Two-path Est.: Unknown}

\addplot [color=color3, line width=1.0pt, forget plot]
  table[row sep=crcr]{%
6.06  572.535615966193\\
9.06  208.192078409949\\
12.06  97.185235096122\\
15.06  52.8607147109602\\
18.06  31.9428623122642\\
21.06  20.8521142363403\\
24.06  14.4327526428264\\
27.06  10.4603444198824\\
30.06  7.89157150183657\\
33.06  6.19989672794677\\
36.06  5.09811632776707\\
39.06  4.41860679829764\\
42.06  4.05441197409359\\
45.06  3.9248460193052\\
48.06  3.95420370342059\\
51.06  4.05784604303883\\
54.06  4.1367451643763\\
57.06  4.09624724707402\\
60.06  3.89426876344251\\
63.06  3.56950359050916\\
66.06  3.20503054864732\\
69.06  2.87279065566138\\
72.06  2.61234027986547\\
75.06  2.4368779534788\\
78.06  2.34341768092841\\
81.06  2.31842473617541\\
84.06  2.33957279015978\\
87.06  2.3764587202262\\
90.06  2.39420932313258\\
93.06  2.36373524318824\\
96.06  2.27539805111076\\
99.06  2.14365250695385\\
102.06  1.99666763683137\\
105.06  1.8614492263156\\
108.06  1.75582092470612\\
111.06  1.68760335729355\\
114.06  1.65639506268564\\
117.06  1.65515076915428\\
120.06  1.67122071955335\\
123.06  1.68789650988409\\
126.06  1.68804316582429\\
129.06  1.6601164143471\\
132.06  1.60337521087885\\
135.06  1.5276900941981\\
138.06  1.44790754719813\\
141.06  1.37756230973355\\
144.06  1.32550297323444\\
147.06  1.29524349588718\\
150.06  1.28544300271531\\
};
%\addlegendentry{Two-path Est.: Known}

\addplot [color=color4, line width=1.0pt, forget plot]
  table[row sep=crcr]{%
6.06  42.4088323795783\\
9.06  28.36541768353\\
12.06  21.309154125973\\
15.06  17.064241410911\\
18.06  14.2296301373535\\
21.06  12.2026143958931\\
24.06  10.6810963280148\\
27.06  9.49694829330229\\
30.06  8.54916033699014\\
33.06  7.77338711241003\\
36.06  7.12669621535839\\
39.06  6.57934501388375\\
42.06  6.11007669948678\\
45.06  5.70329542718743\\
48.06  5.34729946316572\\
51.06  5.03313719838667\\
54.06  4.75384406905673\\
57.06  4.50392019837012\\
60.06  4.27896458892388\\
63.06  4.07541373310687\\
66.06  3.89035144676163\\
69.06  3.72136826859791\\
72.06  3.56645598067834\\
75.06  3.42392742352332\\
78.06  3.29235480166018\\
81.06  3.17052168856096\\
84.06  3.05738530915408\\
87.06  2.95204662043993\\
90.06  2.85372637210305\\
93.06  2.76174579756835\\
96.06  2.6755109231364\\
99.06  2.59449972839266\\
102.06  2.51825157108093\\
105.06  2.44635842368905\\
108.06  2.37845757000468\\
111.06  2.31422548520556\\
114.06  2.25337268183072\\
117.06  2.19563934824641\\
120.06  2.14079164099643\\
123.06  2.0886185195405\\
126.06  2.03892903285516\\
129.06  1.99154998456839\\
132.06  1.94632391628427\\
135.06  1.90310735970379\\
138.06  1.86176931661118\\
141.06  1.822189932835\\
144.06  1.78425933789696\\
147.06  1.74787662668961\\
150.06  1.71294896328189\\
};
%\addlegendentry{\ac{los} Est.: on Two-path model}

\addplot [color=color1, dashed, line width=1.0pt, mark=o, mark options={solid, color1}]
  table[row sep=crcr]{%
6.06  54.9857609590163\\
21.06  19.0107502252181\\
36.06  7.73031897305468\\
51.06  5.24645388005201\\
66.06  4.44299972271688\\
81.06  3.29688696340801\\
96.06  3.09671707757901\\
111.06  2.46613300260062\\
126.06  2.10798246824792\\
141.06  1.85640084346106\\
};
\addlegendentry{Two-path Est.~\eqref{eq:1way_2path_Estimator}: unknown map }

\addplot [color=color3, dashed, mark=square, line width=1.0pt,mark options={solid}]
  table[row sep=crcr]{%
6.06  52.9070503565476\\
21.06  19.1719283173084\\
36.06  5.08601081663575\\
51.06  4.20447046904018\\
66.06  2.9035995793008\\
81.06  2.26708009312136\\
96.06  2.36292480714369\\
111.06  1.70072862924439\\
126.06  1.68512423930371\\
141.06  1.33088668114619\\
};
\addlegendentry{Two-path Est.~\eqref{eq:1way_2path_Estimator}: known map}

\addplot [color=color2, dashed, line width=1.0pt,mark=diamond, mark options={solid, color2}]
  table[row sep=crcr]{%
6.06  102.6049500256417\\
21.06  73.7177471211977\\
36.06  29.2261081064581\\
51.06  19.7535020046353\\
66.06  15.2372226828439\\
81.06  12.3585946155433\\
96.06  10.4930851837764\\
111.06  9.04666033670334\\
126.06  7.77415498991754\\
141.06  7.0029067776267\\
};
\addlegendentry{\ac{los} Est.~\eqref{eq:1way_LoS_Estimator}: on two-path model~\eqref{eq:y_AB_1R}}

\addplot [color=color4, dashed, line width=1.0pt, mark=triangle, mark options={solid, color4}]
  table[row sep=crcr]{%
6.06  40.711614155046\\
21.06  11.3542633801026\\
36.06  7.25471017301888\\
51.06  5.12145325609924\\
66.06  3.7674639708895\\
81.06  2.98744006195163\\
96.06  2.90114332773962\\
111.06  2.34464599209946\\
126.06  2.01008249351385\\
141.06  1.7469837104274\\
};
\addlegendentry{\ac{los} Est.~\eqref{eq:1way_LoS_Estimator}: on \ac{los} model~\eqref{eq:y_A_B_vector}}

% \draw[-Stealth] (axis cs: 4.5, 7.5)   -- (axis cs:4.5, 3.5); 
% \draw[-Stealth] (axis cs: 4, 2)   -- (axis cs: 4.5, 2.6); 

\end{axis}
\end{tikzpicture}%
    \caption{RMSE of the phase offset estimate $\delta_{\phi_{\text{AB}}}$ in degree versus the bandwidth for various estimators for the scenario with known AP positions. }
    \label{fig:PhaseOff_vs_BW_knownPos}
\end{figure}
\begin{figure}[t]
    \centering
    \input{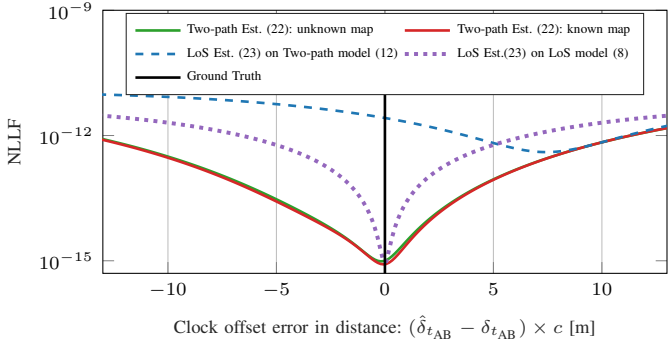}
 \caption{1-D snapshot of NLLF results for clock offset error in distance using 3-D and 1-D search for two-path and LoS estimators, given known AP positions. Signal bandwidth $W=12$ MHz. }
    \label{fig:NLLF_knownPos}
    \vspace{-0.3cm}
\end{figure}
\subsubsection{Impact of Bandwidth}
The \acp{rmse} on the estimation of clock offset $\hat{\delta}_{t_\text{AB}}$ and phase offset $\hat{\delta}_{\phi_\text{AB}^{ }}$ versus bandwidth $W$ are shown in Fig.~\ref{fig:ClockOff_vs_BW_knownPos} and Fig.~\ref{fig:PhaseOff_vs_BW_knownPos}, averaged over 50 Monte Carlo simulations. The corresponding \acp{crlb} are given in solid lines. The results indicate that the proposed ML-based estimators can nearly reach the corresponding CRLBs. Fine-tuning the grid search size can further improve their stability. With unknown map, the reflection path degrades the two-path estimator's performance compared to only LoS estimator on LoS path model. This degradation lessens with growing bandwidths, which allows a better resolution between LoS and reflection paths (delay difference in Hz: $1/(\tau_\text{AR}-\tau_{\text{AB}}) \approx 17$ MHz). With known map, the two-path estimator performs better than the LoS estimator for larger bandwidths ($>30$ MHz), showing a nonlinear improvement with bandwidth. The \ac{los} estimator performs poorly under the two-path model due to mismatched estimators. The two-path estimators' phase offset \acp{rmse} deviates from the \ac{crlb} for small bandwidths ($<20$ MHz) because of phase wrapping. The RMSE values in Fig.~\ref{fig:PhaseOff_vs_BW_knownPos} are larger than those in Fig.~\ref{fig:ClockOff_vs_BW_knownPos}, aligning with~\eqref{eq:CRB_1Path_delta_tau_AB} and~\eqref{eq:CRB_1Path_delta_phi_AB}.

\subsubsection{NLLF Results}
Fixing the bandwidth $W=12$ MHz, Fig.~\ref{fig:NLLF_knownPos} shows the clock offset estimation error, $(\hat{\delta}_{t_{\text{AB}}} - \delta_{t_{\text{AB}}})\times c$ measured in distance, which is a 1D slice of the 3D NLLF results, fixing the other two parameters, for the two-path estimator~\eqref{eq:1way_2path_Estimator} and the 1D NLLF of the LoS estimator~\eqref{eq:1way_LoS_Estimator}. The LoS estimator~\eqref{eq:1way_LoS_Estimator} shows inferior performance under the two-path observation model~\eqref{eq:y_AB_1R} due to interference from the reflection path. In contrast, the two-path estimator~\eqref{eq:1way_2path_Estimator} with unknown reflection parameters offers better accuracy. Having knowledge of reflection path parameters offers almost no advantage because of the limited bandwidth of $12$ MHz. The LoS estimator is most effective under the LoS observation model~\eqref{eq:y_A_B_vector}, displaying the sharpest NLLF curve near the ground truth. The results demonstrate the superior performance of the proposed estimators, highlighting the importance of a precise multi-path model over map data for calibration accuracy.

\subsection{Scenario 2: Calibration with Unknown AP Positions}

%\subsubsection{Impact of Bandwidth}
Fig.~\ref{fig:ClockOff_Delay_vs_BW_UnknownPos} shows the \acp{rmse} on the estimation of clock offset $\delta_{\tau_{\text{AB}}}$ and delay $\tau_{\text{AB}}$ versus bandwidth $W$ for the LoS estimator~\eqref{eq:2way_LoS_Estimator} under unknown AP positions and bi-directional observations over LoS and two-path channel. The results of the LoS estimator~\eqref{eq:1way_LoS_Estimator} using uni-directional observation from Fig.~\ref{fig:ClockOff_vs_BW_knownPos} is also shown. The corresponding \acp{crlb} are given in solid lines. 

As bandwidth increases, the performance of both LoS estimators for clock offset $\delta_{\tau_{\text{AB}}}$ improves linearly. The bi-directional LoS estimator~\eqref{eq:2way_LoS_Estimator} halves the estimation error, as discussed by the CRLB analysis in Section~\ref{section:CRLB_unknownPos}. The delay estimation $\tau_{\text{AB}}$ demonstrates a threshold effect where, with sufficiently large bandwidth ($>100$ MHz), the estimation can approach its \ac{crlb} at a much lower level (see right y-axis) due to carrier phase exploitation, similar to the threshold effect discussed in~\cite{wymeersch2023fundamental}. The two-path measurements have a negligible effect on the bi-directional LoS estimator~\eqref{eq:2way_LoS_Estimator} for both clock offset and delay, in contrast to the significant impact on the uni-directional LoS estimator~\eqref{eq:1way_LoS_Estimator}, as illustrated in Fig.~\ref{fig:ClockOff_vs_BW_knownPos} and Fig.~\ref{fig:NLLF_knownPos}, mainly because bi-directional measurements help to decouple reflection parameters from LoS parameters. For instance, $\delta_{\phi_{\text{AB}}}$ changes signs in~\eqref{eq:y_AB_1R} and~\eqref{eq:y_BA_1R}, unlike $\tau_{\text{AR}}$ and $\delta_{t_{\text{AB}}}$.

This threshold effect arises due to the integer ambiguity error~\eqref{eq:est_pha_off_add_Bidrect} in the phase offset estimation when exploiting the carrier phase. Fig.~\ref{fig:NLLF_unknownPos} shows the 1-D extraction from the 2-D \ac{nllf} results of the estimator~\eqref{eq:2way_LoS_Estimator} for delay estimation at 50 MHz and 200 MHz. With larger bandwidths, a distinct global optimum correctly identifies the optimum, while smaller bandwidths lead to numerous local optima and incorrect estimates. Mismatched two-path measurements on the LoS estimator show that the large bandwidth advantage is less apparent compared to the cases without mismatch.

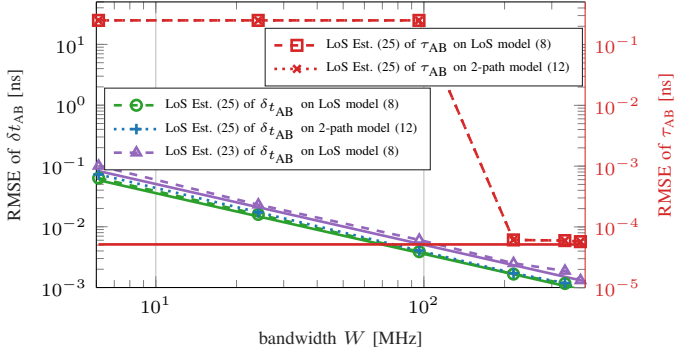
\begin{figure}[t]
    \centering
\begin{tikzpicture}[font=\scriptsize]
\definecolor{color2}{rgb}{0.12156862745098,0.466666666666667,0.705882352941177}
\definecolor{color0}{rgb}{1,0.498039215686275,0.0549019607843137}
\definecolor{color1}{rgb}{0.172549019607843,0.627450980392157,0.172549019607843}
\definecolor{color3}{rgb}{0.83921568627451,0.152941176470588,0.156862745098039}
\definecolor{color4}{rgb}{0.580392156862745,0.403921568627451,0.741176470588235}
\definecolor{color5}{rgb}{0.549019607843137,0.337254901960784,0.294117647058824}
\definecolor{color6}{rgb}{0.890196078431372,0.466666666666667,0.76078431372549}
\definecolor{color7}{rgb}{0.737254901960784,0.741176470588235,0.133333333333333}

\begin{axis}[%
width=6.5cm,
height=3.8 cm,
at={(0,0)},
axis y line*=left,
scale only axis,
xmode=log,
xmin=6,
xmax=400,
xlabel style={font=\color{white!15!black}},
xlabel={\scriptsize bandwidth $W$ [MHz] },
ymode=log,
ymin=1e-3,
ymax=50,
ylabel style={font=\color{white!15!black}},
ylabel={\scriptsize RMSE of $\delta{t_{\text{AB}}}$ [ns]},
axis background/.style={fill=white},
title style={font=\bfseries},
xmajorgrids,
%xminorgrids,
%minor xtick={1, 2, 3, 4, 5, 6, 7, 8, 9, 10,20,30,40,50},
%minor x tick num=9,
%ymajorgrids,
legend style={at={(0.35,0.7)}, anchor=north, legend cell align=left, draw=white!15!black,font=\tiny} % Set legend font size to tiny
,
legend columns=1,
%ytick distance={1e-2},
%xtick distance={10},
% ytick={1e-4,1e-3,1e-2,1e-1}, % Custom tick positions
% yticklabels={$10^0$,$10^1$,20,50} % Custom tick labels]
]

\addplot [color=color1, line width=1.0pt, mark options={solid, color1},forget plot]
  table[row sep=crcr]{%
6.12  0.0583320271612078\\
24.12  0.0147980012800896\\
96.12  0.00371331306375712\\
384.12  0.000929197623019377\\
1536.12  0.000232353836045488\\
};
%\addlegendentry{CRLB (Bi-Direction) of $\delta_{t_{\text{AB}}}$}

\addplot [color=color3, line width=1.0pt, mark options={solid, color3},forget plot]
  table[row sep=crcr]{%
6.12  5.15174313561371e-05\\
24.12  5.15171392550135e-05\\
96.12  5.15124941240318e-05\\
384.12  5.14384532647437e-05\\
1536.12  5.02960147245628e-05\\
};
%\addlegendentry{CRLB (Bi-Direction) of $\tau_{\text{AB}}: f_c=$ 2GHz}

\addplot [color=color4, line width=1.0pt, mark options={solid, color4},forget plot]
  table[row sep=crcr]{%
6.12  0.0824939439111295\\
24.12  0.0209275341064605\\
96.12  0.00525141769610069\\
384.12  0.00131408388059855\\
1536.12  0.000328597946204944\\
};
%\addlegendentry{CRLB (1-Direction) of $\delta_{t_{\text{AB}}}$: known ${\tau_{\text{AB}}}$}

\addplot [color=color1, dashed, line width=1.0pt, mark=o, mark options={solid, color1}]
  table[row sep=crcr]{%
6.12  0.0632061593397136\\
24.12  0.0159632294230248\\
96.12  0.00388999877556156\\
216  0.00166403540619204\\
336  0.00116483067167576\\
384.12  0.000907705397155643\\
};
\addlegendentry{LoS Est.~\eqref{eq:2way_LoS_Estimator} of $\delta_{t_{\text{AB}}}$ on LoS model~\eqref{eq:y_A_B_vector} }

\addplot [color=color2, dotted, line width=1.0pt, mark=+, mark options={solid, color2}]
  table[row sep=crcr]{%
6.12  0.0726737592364927\\
24.12  0.017357932579272\\
96.12  0.004057925296527\\
216  0.0017257925729572\\
336  0.00118529526956295\\
384.12  0.000922592759275927\\
};
\addlegendentry{LoS Est.~\eqref{eq:2way_LoS_Estimator} of $\delta_{t_{\text{AB}}}$ on 2-path model~\eqref{eq:y_AB_1R}}

\addplot [color=color4, dashed, line width=1.0pt, mark=triangle, mark options={solid, color4}]
  table[row sep=crcr]{%
6.12  0.100380380268833\\
24.12  0.0230678255155171\\
96.12  0.00607912070046775\\
216  0.00250537889339227\\
336  0.0018769852452763\\
384.12  0.00131168129971055\\
};
\addlegendentry{LoS Est.~\eqref{eq:1way_LoS_Estimator} of $\delta_{t_{\text{AB}}}$ on LoS model~\eqref{eq:y_A_B_vector}}

% \draw[-Stealth] (axis cs: 4.5, 7.5)   -- (axis cs:4.5, 3.5); 
% \draw[-Stealth] (axis cs: 4, 2)   -- (axis cs: 4.5, 2.6); 

\end{axis}

\begin{axis}[%
width=6.5cm,
height=3.8 cm,
at={(0,0)},
axis y line*=right,
axis x line=none,
scale only axis,
ylabel={RMSE of ${\tau_{\text{AB}}}$ [ns]},
yticklabel style={color=color3}, % Color of the y-axis ticks
y label style={color=color3}, % Color of the y-axis label
y axis line style={color3}, % Color of the y-axis line
ymode=log,
xmode=log,
ymin=1e-5,
ymax=0.5,
xmin=6,
xmax=400,
%ytick distance={2},
%xtick distance={10},
%xtick={1,10,20,50}, % Custom tick positions
%xticklabels={$10^0$,$10^1$,20,50} % Custom
legend style={at={(0.67,0.91)}, anchor=north, legend cell align=left, draw=white!15!black,font=\tiny}
]

\addplot [color=color3, dashed, line width=1.0pt, mark=square, mark options={solid, color3}]
  table[row sep=crcr]{%
6.12  0.249982930362031\\
24.12  0.24998276242256\\
96  0.24994811515936\\
216  6.10982223520846e-05\\
336  5.95592577438793e-05\\
384.12  5.72223669375023e-05\\
};
\addlegendentry{LoS Est.~\eqref{eq:2way_LoS_Estimator} of $\tau_{\text{AB}}$ on LoS model~\eqref{eq:y_A_B_vector}}

\addplot [color=color3, dotted, line width=1.0pt, mark=x, mark options={solid, color3}]
  table[row sep=crcr]{%
6.12  0.24992795629655\\
24.12  0.2499520527052\\
96  0.249825027502226\\
216  6.11265280627256e-05\\
336  5.952572052725252e-05\\
384.12  5.7225926592752e-05\\
};
\addlegendentry{LoS Est.~\eqref{eq:2way_LoS_Estimator} of $\tau_{\text{AB}}$ on 2-path model~\eqref{eq:y_AB_1R}}

\addplot [color=color3, line width=1.0pt, mark options={solid, color3},forget plot]
  table[row sep=crcr]{%
6.12  5.15174313561371e-05\\
24.12  5.15171392550135e-05\\
96.12  5.15124941240318e-05\\
384.12  5.14384532647437e-05\\
1536.12  5.02960147245628e-05\\
};
%\addlegendentry{CRLB (Bi-Direction) of $\tau_{\text{AB}}: f_c=$ 2GHz}
\end{axis}
\end{tikzpicture}%
    \caption{\ac{rmse} on the estimation of the clock offset $\hat{\delta}_{t_{\text{AB}}}$ (left y-axis) and delay $\hat{\tau}_{\text{AB}}$ (right y-axis) versus the bandwidth for two LoS estimators, \eqref{eq:1way_LoS_Estimator} and~\eqref{eq:2Path_est_knownPos}, in scenarios with and without known AP positions and LoS path. }
\label{fig:ClockOff_Delay_vs_BW_UnknownPos}
\vspace{-0.cm}
\end{figure}

\begin{figure}[t]
    \centering
\input{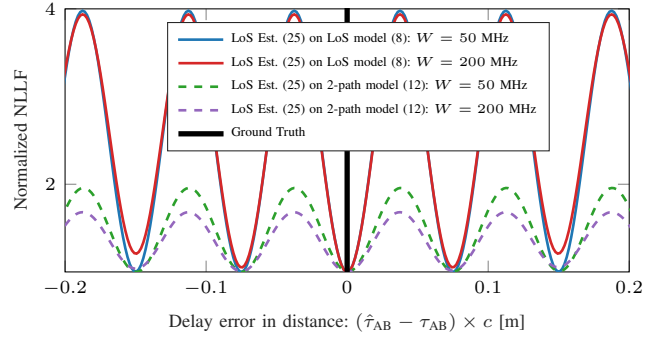}
    \caption{1-D snapshot of 2-D NLLF results for delay estimate error in distance for the LoS estimators and~\eqref{eq:2way_LoS_Estimator}.}
\label{fig:NLLF_unknownPos}
\end{figure}

\section{Conclusion}
This paper addresses phase and time calibration in distributed mMIMO context using a wideband multi-path model for separate phase and time offsets calibration. We cover scenarios with/without known AP positions and/or reflection information, offering a practical framework. We developed ML estimators for each scenario. Our CRLB analyses establish theoretical performance benchmarks, improving the understanding of calibration limits. Extensive simulations validated our estimators against the CRLBs. Focusing on two single-antenna APs, we note that our model and techniques are applicable to intra-AP and AP-UE calibration in distributed massive MIMO. Future work can explore calibration with unknown AP positions in a multi-path channel and localization with multiple APs.

%%%%%%%%%%%%%%%%%%%%%%%%%%%%%%%%%%%%%%%%%
\appendices
\vspace{-0.1cm}
\section{Proof of \eqref{eq:2way_2path_Estimator} in Proposition 1} \label{sec:Appendix_TwoPath_unknown}
\vspace{-0.1cm}
Using the \ac{ml} criterion, we solve the estimation of $\boldsymbol{\eta}$ from \eqref{eq:unknownPar_unknown_positions_LoS} using the bi-directional two-path observation $\boldsymbol{y}_{\text{}}=[\boldsymbol{y}_{\text{AB}}^{\mathsf{T}},\boldsymbol{y}_{\text{BA}}^{\mathsf{H}}]^{\mathsf{T}} \in \mathbb{C}^{2N}$ from~\eqref{eq:y_AB_1R} and~\eqref{eq:y_BA_1R} as
\begin{align}
    \hat{\boldsymbol{\eta}}^{\text{}} = \arg \underset{\boldsymbol{\eta}_{}}{\max} \quad p(\boldsymbol{y}_{} | \boldsymbol{\eta}_{}).\label{eq:obj_func_p_y_Bidirect_2path}
\end{align}
Denote $\boldsymbol{\mu}_{\text{AB-BA}}\triangleq [\boldsymbol{\mu}_{\text{AB}}^{\mathsf{T}}, \boldsymbol{\mu}_{\text{BA}}^{\mathsf{H}}]^{\mathsf{T}}\in \mathbb{C}^{2N} $, and $\boldsymbol{\mu}_{\text{R}} \triangleq [\boldsymbol{\mu}_{\text{AR}}^{\mathsf{T}}, \boldsymbol{\mu}_{\text{BR}}^{\mathsf{H}}]^{\mathsf{T}} \in \mathbb{C}^{2N}$. Solving~\eqref{eq:obj_func_p_y_Bidirect_2path} is equivalent to minimize the negative log-likelihood version of  $ p(\boldsymbol{y}_{} | \boldsymbol{\eta}_{})$. After removing irrelevant terms in $ \log p(\boldsymbol{y}_{} | \boldsymbol{\eta}_{})$, the problem in~\eqref{eq:obj_func_p_y_Bidirect_2path} becomes to minimize the \ac{nllf}, given by
% \begin{align}
    $\hat{\boldsymbol{\eta}}^{\text{}} = \arg \underset{\boldsymbol{\eta}_{}}{\min} \quad  \mathcal{L}_{\text{ML}}(\boldsymbol{\eta}_{})$,
    %\label{eq:arg_ML_bidirect_2Path}
% \end{align}
where
\begin{align}
   &\mathcal{L}_{\text{ML}}(\boldsymbol{\eta}_{}) \triangleq
    \left\|
   \boldsymbol{y}_{} - \boldsymbol{\mu}_{\text{AB-BA}} - \boldsymbol{\mu}_{\text{R}} 
    \right\|^2_2 \nonumber \\ &= \Big\|
   \boldsymbol{y}_{} - \beta_{\text{AB}} [e^{-j{\delta_{\phi_{\text{AB}}^{}}}} e^{j{\tilde{\varphi}^{}_{\text{AB}}}}  \boldsymbol{c}_{\text{AB}}^{\mathsf{T}},  e^{j{\delta_{\phi_{\text{BA}}^{}}}} e^{-j{\tilde{\varphi}^{}_{\text{BA}}}}  \boldsymbol{c}_{\text{BA}}^{\mathsf{H}}]^{\mathsf{T}}  \nonumber \\
& \hspace{0.8cm} -\beta_{\text{AR}} [e^{-j{\delta_{\phi_{\text{AB}}^{}}}} e^{j{\tilde{\varphi}^{}_{\text{AR}}}}  \boldsymbol{c}_{\text{AR}}^{\mathsf{T}},  e^{j{\delta_{\phi_{\text{BA}}^{}}}} e^{-j{\tilde{\varphi}^{}_{\text{BR}}}}  \boldsymbol{c}_{\text{BR}}^{\mathsf{H}}]^{\mathsf{T}} 
    \Big\|^2_2, \label{eq:obj_func_Bidirect_log_p_before_2path}\\
    &=\Big\|
\boldsymbol{y} - \beta_{\text{AB}} e^{-j{\delta_{\phi_{\text{AB}}^{}}}} \boldsymbol{c}_{\text{AB-BA}}
-\beta_{\text{AR}} e^{-j{\delta_{\phi_{\text{AB}}^{}}}} \boldsymbol{c}_{\text{AR-BR}}
    \Big\|^2_2,
    \label{eq:obj_func_Bidirect_log_p_2path}
\end{align}
where $\tilde{\varphi}_{\text{AB}} \triangleq \varphi_{\text{AB}} - \delta_{\phi_{\text{AB}}^{}}$, $\tilde{\varphi}_{\text{BA}} \triangleq \varphi_{\text{BA}} - \delta_{\phi_{\text{BA}}^{}}$, $\tilde{\varphi}_\text{AR} \triangleq \varphi_\text{AR} - \delta_{\phi_{\text{AB}}^{}}$, $\tilde{\varphi}_\text{BR} \triangleq \varphi_\text{BR} - \delta_{\phi_{\text{BR}}^{}}$, $\boldsymbol{c}_{\text{AB-BA}} \triangleq [e^{j{\tilde{\varphi}^{}_{\text{AB}}}}  \boldsymbol{c}_{\text{AB}}^{\mathsf{T}}, e^{-j{\tilde{\varphi}^{}_{\text{BA}}}}  \boldsymbol{c}_{\text{BA}}^{\mathsf{H}}]^{\mathsf{T}}$, $\boldsymbol{c}_{\text{AR-BR}} \triangleq [e^{j{\tilde{\varphi}^{}_{\text{AR}}}}  \boldsymbol{c}_{\text{AR}}^{\mathsf{T}}, e^{-j{\tilde{\varphi}^{}_{\text{BR}}}}  \boldsymbol{c}_{\text{BR}}^{\mathsf{H}}]^{\mathsf{T}}$, and from \eqref{eq:obj_func_Bidirect_log_p_before_2path} to \eqref{eq:obj_func_Bidirect_log_p_2path} we use $\delta_{\phi_{\text{AB}}^{}}=-\delta_{\phi_{\text{BA}}^{}}$ (defined in~\eqref{eq:delta_phi_AB_BA}). 

We can express two channel gains in vector form as $\boldsymbol{\beta} = [\beta_{\text{AB}},\beta_\text{AR}]^{\mathsf{T}}$. Thus we can rewrite~\eqref{eq:obj_func_Bidirect_log_p_2path} as
\begin{align}
   \mathcal{L}_{\text{ML}}(\boldsymbol{\eta}_{}) = \Big\|
   \boldsymbol{y}_{\text{}} - e^{-j\delta_{\phi_{\text{AB}}^{}}}\boldsymbol{C}_{\text{A-R-B}} \boldsymbol{\beta},
    \Big\|^2_2,\label{eq:obj_func_log_p_vector_Bidirect}
\end{align}
where $\boldsymbol{C}_{\text{A-R-B}}\triangleq [\boldsymbol{c}_{\text{AB-BA}}, \boldsymbol{c}_\text{AR-BR}] \in \mathbb{C}^{2N\times 2}$. To solve this ML estimation problem, we first estimate $\boldsymbol{\beta}$ as a function of the remaining parameters in closed-form as
\begin{align}
    \hat{\boldsymbol{\beta}}_{}^{\text{ML}} &=  \Re\big\{ \big(\boldsymbol{C}_{\text{A-R-B}}^{\mathsf{H}} \boldsymbol{C}_{\text{A-R-B}} 
    \big)^{-1} (e^{-j\delta_{\phi_{\text{AB}}^{}}} \boldsymbol{C}_{\text{A-R-B}})^{\mathsf{H}} \boldsymbol{y}_{\text{}} \big\} \nonumber \\
    &= \frac{1}{2} \big(\boldsymbol{C}_{\text{A-R-B}}^{\mathsf{H}} \boldsymbol{C}_{\text{A-R-B}} 
    \big)^{-1} (e^{-j\delta_{\phi_{\text{AB}}^{}}} \boldsymbol{C}_{\text{A-R-B}})^{\mathsf{H}} \boldsymbol{y}_{\text{}} \nonumber\\
   & \hspace{0.3cm} + \frac{1}{2}\Big( \boldsymbol{y}_{\text{}}^{\mathsf{H}} (e^{-j\delta_{\phi_{\text{AB}}^{}}} \boldsymbol{C}_{\text{A-R-B}}) \big(\boldsymbol{C}_{\text{A-R-B}}^{\mathsf{H}} \boldsymbol{C}_{\text{A-R-B}} 
    \big)^{-1}   
    \Big)^{\mathsf{T}}.
    \label{eq:ML_est_beta_vector_Bidirect}
\end{align}

Substituting~\eqref{eq:ML_est_beta_vector_Bidirect} into~\eqref{eq:obj_func_log_p_vector_Bidirect}  drops the dependency of $\boldsymbol{\beta}$, we obtain the compressed loss function
\begin{align}
    &\mathcal{L}_{\text{ML}}(\tau_{\text{AB}}, \delta_{\phi_{\text{AB}}^{}},\delta_{t_{\text{AB}}}, \tau_{\text{AR}}, \delta_{\phi_\text{AR}^{ }}) = \| \breve{\boldsymbol{y}}_{\text{}} - e^{-2j\delta_{\phi_{\text{AB}}^{}}} \breve{\boldsymbol{c}}_{\text{A-R-B}}   
    \|^2_2, \label{eq:obj_func_compres_beta_Bidirect}
\end{align} 
where $\breve{\boldsymbol{y}}_{\text{}} \triangleq {\boldsymbol{y}_{\text{}} -
    \frac{1}{2} \boldsymbol{C}_{\text{A-R-B}} \big(\boldsymbol{C}_{\text{A-R-B}}^{\mathsf{H}} \boldsymbol{C}_{\text{A-R-B}} 
    \big)^{-1} \boldsymbol{C}_{\text{A-R-B}}^{\mathsf{H}} \boldsymbol{y}_{\text{}}}$ and $ 
    \breve{\boldsymbol{c}}_{\text{A-R-B}} \triangleq {\frac{1}{2} \boldsymbol{C}_{\text{A-R-B}} \big(\boldsymbol{C}_{\text{A-R-B}}^{\mathsf{H}} \boldsymbol{C}_{\text{A-R-B}} 
    \big)^{-1} \boldsymbol{C}_{\text{A-R-B}}^{\mathsf{T}} (\boldsymbol{y}_{\text{}}^{\mathsf{H}})^T}$.
Similarly, we can also estimate $\delta_{\phi_{\text{AB}}^{}}$ in closed-form using the remaining parameters as
\begin{align}
   \hat{\delta}_{\phi^{ }_{\text{AB}}} =  -\frac{\angle \big(\breve{\boldsymbol{c}}^{\mathsf{H}}_{\text{A-R-B}}\breve{\boldsymbol{y}}_{\text{}}\big)}{2} + z_1\pi, 
   \label{eq:est_pha_off_add_Bidrect}
\end{align}
where $z_{1} \in \mathbb{Z}$ is introduced to account for possible integer ambiguities in phase estimation. Inserting~\eqref{eq:est_pha_off_add_Bidrect} into~\eqref{eq:obj_func_compres_beta_Bidirect} yields
\begin{align}
    \mathcal{L}_{\text{2-path}}^{\text{Bi}}(\tau_{\text{AB}},\delta_{t_{\text{AB}}},\tau_{\text{AR}}, \delta_{\phi_\text{AR}^{ }}) = &\left\| \breve{\boldsymbol{y}}_{}\right\|^2_2 + \left\| \breve{\boldsymbol{c}}_{\text{ABA}}\right\|^2_2 
     - 2 \left| \breve{\boldsymbol{c}}_{\text{ABA}}^{\mathsf{H}} \breve{\boldsymbol{y}}_{}\right|.
     \label{eq:2Path_est_unknownPos}
\end{align}

%%%%%%%%%%%%%%%%%%%%%%%%%%%%%%%%%%%%%%%%%%%%%%%%%%%%%%%%%%%
\section{Proof of \eqref{eq:1way_2path_Estimator}, \eqref{eq:1way_LoS_Estimator}, and \eqref{eq:2way_LoS_Estimator} in Proposition 1}\label{sec:Appendix_Other3Estimators}
Using the ML criterion, we solve the estimation of $\boldsymbol{\eta}$ from~\eqref{eq:unknownPar_known_positions},~\eqref{eq:unknownPar_known_positions_LoS}, and~\eqref{eq:unknownPar_unknown_positions_LoS} by minimizing the following \acp{nllf}

\noindent$\mathcal{L}_{\text{ML}}(\boldsymbol{\eta}) \triangleq$
\begin{numcases}{}
    \left\|\boldsymbol{y}_{\text{AB}} - \boldsymbol{\mu}_{\text{AB}} - \boldsymbol{\mu}_\text{AR}\right\|^2_2, & \text{for} $\boldsymbol{\eta}$ \text{in} \eqref{eq:unknownPar_known_positions}, $\boldsymbol{y}_{\text{AB}}$ \text{in} \eqref{eq:y_AB_1R} \label{eq:NLLF_collect_1}\\
    \left\|\boldsymbol{y}_{\text{AB}} - \boldsymbol{\mu}_{\text{AB}}\right\|^2_2, & \text{for} $\boldsymbol{\eta}$ \text{in} \eqref{eq:unknownPar_known_positions_LoS}, $\boldsymbol{y}_{\text{AB}}$ \text{in} \eqref{eq:y_A_B_vector} \label{eq:NLLF_collect_2}\\
    \left\|\boldsymbol{y} - \boldsymbol{\mu}_{\text{AB-BA}}\right\|^2_2, & \text{for} $\boldsymbol{\eta}$ \text{in} \eqref{eq:unknownPar_unknown_positions_LoS}, $\boldsymbol{y}$ \text{in} \eqref{eq:y_A_B_vector}, \eqref{eq:y_B_A_vector}, \label{eq:NLLF_collect_3}
\end{numcases}
where the bi-directional LoS observation $\boldsymbol{y}_{\text{}}=[\boldsymbol{y}_{\text{AB}}^{\mathsf{T}},\boldsymbol{y}_{\text{BA}}^{\mathsf{H}}]^{\mathsf{T}} \in \mathbb{C}^{2N}$ in~\eqref{eq:NLLF_collect_3} is obtained from~\eqref{eq:y_A_B_vector} and~\eqref{eq:y_B_A_vector}.

We omit some details for simplicity and follow similar steps from~\eqref{eq:obj_func_Bidirect_log_p_before_2path} to~\eqref{eq:est_pha_off_add_Bidrect} to compress the unknown $\beta_{\text{AB}}$, $\beta_{\text{BA}}$ (if applicable), $\beta_{\text{AR}}$ (if applicable), $\beta_{\text{BR}}$ (if applicable), and $\delta_{\phi_{\text{AB}}}$ in each NLLF function in~\eqref{eq:NLLF_collect_1},~\eqref{eq:NLLF_collect_2}, and~\eqref{eq:NLLF_collect_3}. Specifically, we obtain the compressed loss function of~\eqref{eq:NLLF_collect_1} for the uni-directional two-path estimator,
\begin{align}
    \mathcal{L}_{\text{2-path}}^{\text{Uni}}(\delta_{t_{\text{AB}}}, \tau_\text{AR}, \delta_{\phi_\text{AR}^{ }}) = &\left\| \breve{\boldsymbol{y}}_{\text{AB}}\right\|^2_2 + \left\| \breve{\boldsymbol{c}}_{\text{ABR}}\right\|^2_2 
     - 2 \left| \breve{\boldsymbol{c}}_{\text{ABR}}^{\mathsf{H}} \breve{\boldsymbol{y}}_{\text{AB}}\right|, \label{eq:2Path_est_knownPos}
\end{align}
where $\breve{\boldsymbol{y}}_{\text{AB}} \triangleq {\boldsymbol{y}_{\text{AB}} -
    \frac{1}{2} \boldsymbol{C}_{\text{ABR}} \big(\boldsymbol{C}_{\text{ABR}}^{\mathsf{H}} \boldsymbol{C}_{\text{ABR}} 
    \big)^{-1} \boldsymbol{C}_{\text{ABR}}^{\mathsf{H}} \boldsymbol{y}_{\text{AB}}}$ and $ 
    \breve{\boldsymbol{c}}_{\text{ABR}} \triangleq {\frac{1}{2} \boldsymbol{C}_{\text{ABR}} \big(\boldsymbol{C}_{\text{ABR}}^{\mathsf{H}} \boldsymbol{C}_{\text{ABR}} 
    \big)^{-1} \boldsymbol{C}_{\text{ABR}}^{\mathsf{T}} (\boldsymbol{y}_{\text{AB}}^{\mathsf{H}})^T}$, $\boldsymbol{C}_{\text{ABR}}\triangleq [e^{j{\tilde{\varphi}^{}_{\text{AB}}}}  \boldsymbol{c}_{\text{AB}}, e^{j{\tilde{\varphi}^{}_\text{AR}}}  \boldsymbol{c}_\text{AR}] \in \mathbb{C}^{N\times 2}$, $\boldsymbol{c}_{\text{AB}} \triangleq \boldsymbol{a}(\tilde{\tau}_{\text{AB}}) \odot \boldsymbol{s}_{\text{A}}$ and $ \boldsymbol{c}_\text{AR} \triangleq \boldsymbol{a}(\tilde{\tau}_\text{AR}) \odot \boldsymbol{s}_{\text{A}}$.
    
Similarly, we can also obtain the compressed loss function of~\eqref{eq:NLLF_collect_2} for the uni-directional LoS estimator,
\begin{align}
    \mathcal{L}_{\text{LoS}}^{\text{Uni}}(\delta_{t_{\text{AB}}}) = &\left\| \Acute{\boldsymbol{y}}_{\text{AB}}\right\|^2_2 + \left\| \Acute{\boldsymbol{c}}_{\text{AB}}\right\|^2_2
   - 2 \left| \Acute{\boldsymbol{c}}_{\text{AB}}^{\mathsf{H}} \Acute{\boldsymbol{y}}_{\text{AB}}\right|,\label{eq:NLLF_1way_LoS_Estimator}
\end{align}
where $\Acute{\boldsymbol{y}}_{\text{AB}} \triangleq \boldsymbol{y}_{\text{AB}} - {\boldsymbol{c}_{\text{AB}}^{\mathsf{H}} \boldsymbol{y}_{\text{AB}}  }/{2\| \boldsymbol{c}_{\text{AB}} \|^2_2}\boldsymbol{c}_{\text{AB}}$ and $\Acute{\boldsymbol{c}}_{\text{AB}} \triangleq {e^{-j4\pi f_c \tilde{\tau}_{\text{AB}}} \boldsymbol{y}_{\text{AB}}^{\mathsf{H}} \boldsymbol{c}_{\text{AB}}}{\| \boldsymbol{c}_{\text{AB}} \|^2_2 }  \boldsymbol{c}_{\text{AB}}/2$.
We obtain the compressed loss function of~\eqref{eq:NLLF_collect_3} for the bi-directional LoS estimator,
\begin{align}
    \mathcal{L}_{\text{LoS}}^{\text{Bi}}(\tau_{\text{AB}},\delta_{t_{\text{AB}}}) = &\left\| \tilde{\boldsymbol{y}}_{}\right\|^2_2 + \left\| \tilde{\boldsymbol{c}}_{\text{ABA}}\right\|^2_2 
     - 2 \left| \tilde{\boldsymbol{c}}_{\text{ABA}}^{\mathsf{H}} \tilde{\boldsymbol{y}}_{}\right|, \label{eq:NLLF_1Path_est_unknownPos}
\end{align}
where $\tilde{\boldsymbol{y}} \triangleq \boldsymbol{y} - {\boldsymbol{c}_{\text{ABA}}^{\mathsf{H}} \boldsymbol{y}\boldsymbol{c}_{\text{ABA}} }/{\|\boldsymbol{c}_{\text{ABA}}\|^2_2}$ and $\tilde{\boldsymbol{c}}_{\text{ABA}} \triangleq {\boldsymbol{y}_{}^{\mathsf{H}} \boldsymbol{c}_{\text{ABA}}\boldsymbol{c}_{\text{ABA}} } /{\|\boldsymbol{c}_{\text{ABA}}\|^2_2}$. 

\section{CRLB Calculation of $\boldsymbol{\eta}$ from~\eqref{eq:unknownPar_known_positions}. } \label{appdices: CRLB_1}
The \ac{fim} of $\boldsymbol{\eta}$ from~\eqref{eq:unknownPar_known_positions} is calculated following~\eqref{eq:FIM_Compute} using ${\partial \boldsymbol{\mu}_{}^{\mathsf{}}}/{\partial \delta_{t_{\text{AB}}}} = \boldsymbol{D}_1 \boldsymbol{\mu}$, 
  ${\partial \boldsymbol{\mu}_{}^{\mathsf{}}}/{\partial \delta_{\phi_{\text{AB}}^{}}}= -j\boldsymbol{\mu}$, 
  ${\partial \boldsymbol{\mu}_{}^{\mathsf{}}}/{\partial \tau_{\text{AR}}}= \boldsymbol{D}_1 \boldsymbol{\mu}_\text{AR}$
  ${\partial \boldsymbol{\mu}_{}^{\mathsf{}}}/{\partial \delta_{\phi_{\text{AR}}^{}}} = -j\boldsymbol{\mu}_{\text{AR}}$ 
  ${\partial \boldsymbol{\mu}_{}^{\mathsf{}}}/{\partial \beta_\text{AR}}= \boldsymbol{\mu}_\text{AR}/\beta_\text{AR}$
$  {\partial \boldsymbol{\mu}_{}^{\mathsf{}}}/{\partial \beta_{\text{AB}}} = \boldsymbol{\mu}_{\text{AB}}/\beta_{\text{AB}}$. Here $\boldsymbol{D}_1=\text{diag}(\boldsymbol{d}_1)$, and $\boldsymbol{d}_1 =-j2\pi (f_c \boldsymbol{1}_{N} +\Delta_f ([-(N-1)/2,\cdots,(N-1)/2]))$. Using these derivatives, the corresponding FIM elements can be calculated following~\eqref{eq:FIM_Compute}. Notably, the FIM size is reduced based on prior knowledge of $\tau_{\text{AR}}$ and $\delta_{\phi_\text{AR}^{ }}$. Finally, the error covariance bounds for $\delta_{t_{\text{AB}}}$, $\delta_{\phi_{\text{AB}}^{}}$, $\tau_{\text{AR}}$, and $\delta_{\phi_\text{AR}^{ }}$ (if applicable) are obtained following~\eqref{eq:each_error_bound_general}.

\section{CRLB Calculation of $\boldsymbol{\eta}$ from~\eqref{eq:unknownPar_unknown_positions}. } \label{appdices: CRLB_2}
The FIM of $\boldsymbol{\eta}$ from~\eqref{eq:unknownPar_unknown_positions} is calculated following~\eqref{eq:FIM_Compute} using $\boldsymbol{\mu} = (\boldsymbol{\mu}_{\text{AB-BA}}+\boldsymbol{\mu}_\text{R})$, defined after~\eqref{eq:obj_func_p_y_Bidirect_2path}. The derivatives are $ {\partial \boldsymbol{\mu}_{}^{\mathsf{}}}/{\partial {\tau_{\text{AB}}}} = \boldsymbol{D}_{1,1} \boldsymbol{\mu}_{\text{AB-BA}}$,  $ {\partial \boldsymbol{\mu}_{}^{\mathsf{}}}/{\partial \delta_{t_{\text{AB}}}} = \boldsymbol{D}_{1,2} \boldsymbol{\mu}$, $ 
  {\partial \boldsymbol{\mu}_{}^{\mathsf{}}}/{\partial \delta_{\phi_{\text{AB}}^{}}} = -j\boldsymbol{\mu} $, $
  {\partial \boldsymbol{\mu}_{}^{\mathsf{}}}/{\partial \tau_{\text{AR}}} = \boldsymbol{D}_{1,1} \boldsymbol{\mu}_\text{R}$, $
  {\partial \boldsymbol{\mu}_{}^{\mathsf{}}}/{\partial \delta_{\phi_{\text{AR}}^{}}} = -j\boldsymbol{\mu}_{\text{R}}$, $  
  {\partial \boldsymbol{\mu}_{}^{\mathsf{}}}/{\partial \beta_\text{AR}} = \boldsymbol{\mu}_\text{R}/\beta_\text{AR}$, $  
  {\partial \boldsymbol{\mu}_{}^{\mathsf{}}}/{\partial \beta_{\text{AB}}} = \boldsymbol{\mu}_{\text{AB-BA}}/\beta_{\text{AB}} 
$. Here $\boldsymbol{D}_{1,1}= \text{diag}([\boldsymbol{d_1},-\boldsymbol{d_1}])$, $\boldsymbol{D}_{1,2}= \text{diag}([\boldsymbol{d_1},-\boldsymbol{d_2}])$, and $\boldsymbol{d}_2=-j2\pi (f_c \boldsymbol{1}_{N} -\Delta_f \text{diag}([-(N-1)/2,\cdots,(N-1)/2]))$. These derivatives can be used following~\eqref{eq:FIM_Compute}. Notably, the FIM size are reduced based on prior knowledge of $\tau_{\text{AR}}$ and $\delta_{\phi_\text{AR}^{ }}$. Finally, the error covariance bounds for each parameter in~\eqref{eq:unknownPar_unknown_positions} are obtained following~\eqref{eq:each_error_bound_general}.

\bibliographystyle{IEEEtran}
\bibliography{references}
\end{document}